\documentclass[pdflatex,sn-mathphys-num]{sn-jnl}

\usepackage{float}
\usepackage{graphicx}%
\usepackage{physics}
\usepackage{multirow}%
\usepackage{amsmath,amssymb,amsfonts}%
\usepackage{amsthm}%
\usepackage[title]{appendix}%
\usepackage{xcolor}%
\usepackage{textcomp}%
\usepackage{manyfoot}%
\usepackage{booktabs}%
\usepackage{algorithm}%
\usepackage{algorithmicx}%
\usepackage{algpseudocode}%
\usepackage{listings}%
\usepackage[utf8]{inputenc}
\usepackage[T1]{fontenc}
\usepackage{newtxtext,newtxmath}
\usepackage[scr=rsfso]{mathalfa}
\DeclareMathAlphabet{\mathscr}{U}{rsfso}{m}{n}%
\usepackage{tabularray}
\usepackage{caption}
\usepackage{academicons}
\usepackage{hyperref}
\usepackage{textpos}
\usepackage{tabularray}
\usepackage[justification=centering, singlelinecheck=false]{caption}

\theoremstyle{thmstyleone}%
\newtheorem{theorem}{Theorem}
\newtheorem{Lemma}{Lemma}
\theoremstyle{thmstyletwo}%

\theoremstyle{thmstylethree}%

\begin{document}

\title[Over and Under Doses of Antibiotic Treatment to Bacterial Resistance]
{Modeling the Effects of Over and Under Doses Antibiotic Treatment to Bacterial Resistance in Presence of Immune System}


\author*[1]{\fnm{Uzzwal Kumar} \sur{Mallick}}\email{mallickuzzwal@math.ku.ac.bd}

\author[1]{\fnm{Jobayer} \sur{Ahmed}}
\author[1]{\fnm{Khan Anik} \sur{Islam}}
\author[1]{\fnm{Pulak} \sur{Kundu}}

\affil[1]{\orgdiv{Mathematics Discipline}, \orgname{Khulna University}, \orgaddress{\street{Khulna-9208}, \country{Bangladesh}}}

\abstract{Antibiotic resistance presents a growing global health threat by diminishing the effectiveness of treatments and allowing once-manageable bacterial infections to persist. This study develops and analyzes an optimization-based mathematical model to investigate the impact of varying antibiotic dosages on bacterial resistance, incorporating the role of the immune system. Additionally, to capture the effects of over- and under-dosing, a floor function is newly introduced into the model as a switch function. The model is examined both analytically and numerically. As part of the analytical solution, the validity of the model through the existence and uniqueness theorem, stability at the equilibrium points, and characteristics of equilibrium points in relation to state variables have been investigated. Numerical simulations, performed using the Runge–Kutta 4th order method, reveal that while antibiotics effectively reduce sensitive bacteria, they simultaneously increase resistant strains and suppress immune cell levels. The results also demonstrate that under-dosing antibiotics increases the risk of resistance through bacterial mutation, while over-dosing weakens the immune system by disrupting beneficial microbes. These findings emphasize that improper dosing—whether below or above the prescribed level—can accelerate the development of antibiotic resistance, underscoring the need for carefully regulated treatment strategies that preserve both antimicrobial effectiveness and immune system integrity.}



\maketitle
\section{Introduction}
In recent times, the emergence of antibiotic-resistant bacteria has become a significant public health problem worldwide, highlighting the major challenges to the therapeutic effectiveness of traditional antibiotic treatments. The intricate interaction among bacterial populations, the amount of antibiotics administered, and the human immune system play an essential role in determining antibiotic resistance patterns. Both overdosing and underdosing therapies carry unintended consequences, potentially selecting resistant bacteria and jeopardizing treatment efficacy, leading to the death of the human population.
\begin{figure}[hbt!]
	\centering
	\includegraphics[height=4.5cm,width=9.0cm]{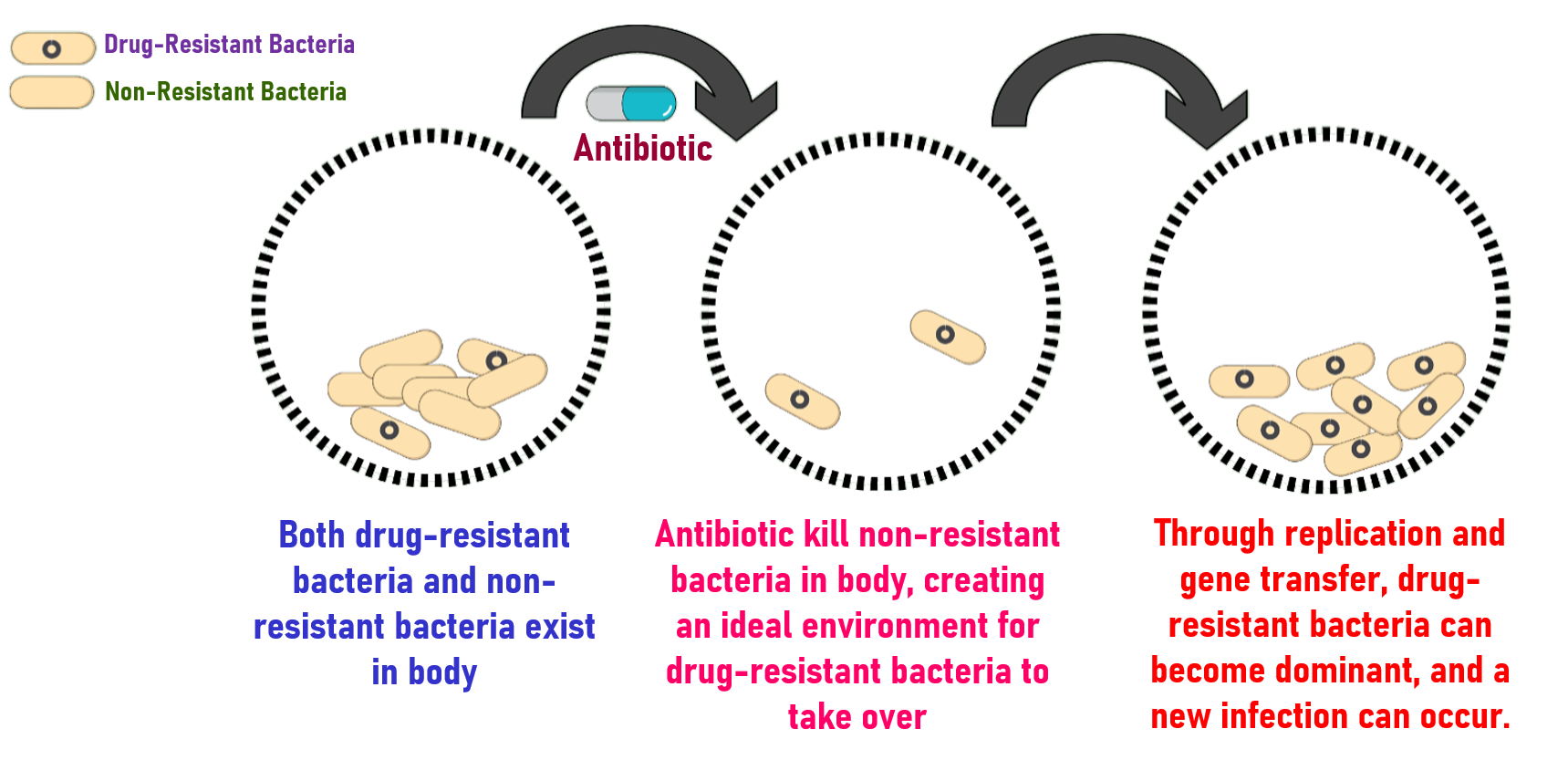}
	\caption{Mechanism of antibiotic in human body}
	\label{fig:antibiotic work}
\end{figure}
Infection has been the leading cause of most diseases throughout human history, with bacterial infections being particularly common \cite{Mondragón2014}. Antibiotic therapy applied to individuals is the most common procedure known to fight bacterial infection. Unfortunately, new bacterial strains resistant to medicines have emerged almost immediately following the introduction of each new class of antibiotics for treatment \cite{Butler2006}. On the other hand, the immune system recognizes and responds to antigens to defend the body against hazardous chemicals. Due to the host's immune system, various hosts might react differently to the same virus \cite{Alberts2002, Palomino2014}.

Antibiotics are some of the biochemical drugs that destroy microorganisms \cite{uddin2021}. In this track, the first antibiotic was discovered in 1928 by the physiologist Alexander Fleming, and on 12 February 1941, the first antibiotic, penicillin, was applied to the human body \cite{Gaynes2017}. Up to now many antibiotics have been invented and are still being discovered today. But, one of the challenges for treating bacterial diseases with antibiotics is the emergence of drug-resistant bacteria \cite{Muteeb2023}. When a person takes antibiotics, sensitive bacteria die, but resistant bacteria may proliferate and reproduce, as shown in Figure \ref{fig:antibiotic work}. By creating an optimal habitat, frequent and incorrect antibiotic usage produces bacteria that are drug-resistant. Bacteria become antibiotic-resistant when they evolve to make medications, chemicals, or other agents less effective \cite{Fqs2019, Baquero2021}. Several processes allow bacteria to achieve this like some bacteria can neutralise antibiotics before they hurt them or through mutations or switch the antibiotic attack location to avoid affecting bacteria function \cite{Kohanski2010, Reygaert2018}. A recent research conducted by BSMMU reveals that the use of antibiotics in Bangladesh has reached a critical level, rendering different kinds of antibiotics ineffective in the human body resulted in the unfortunate death of 26,000 individuals in the last year owing to the diminished efficacy of drugs on the human system \cite{Jamuna2024}. The current utilization of antibiotics from the reserve group of antibiotic will lead to an emergency in which common illnesses, such as the cold, may no longer respond effectively to antibiotics. This raises alarming concerns about the potential dangers to lives caused by apparently slight illnesses \cite{Bonna2022}.
During the duration of an antibiotic prescription, many patients have a tendency to either skip doses or continue taking the medication beyond the recommended intervals prescribed by a doctor \cite{Yusef2017}. Due to selection pressure favoring resistant bacteria, the immune system may struggle to produce an efficient response in the presence of antibiotics that are overused \cite{Bui2023}. Besides, the immune system may become weakened due to the strain of constantly fighting off bacterial infections if antibiotics are not administered often enough \cite{ScienceDaily2023}.
Mathematical modeling plays a vital role in understanding disease dynamics in the human body as well as illustrating the biological phenomenon.  Naaly et al. (2024) \cite{NAALY2024100159} proposed a mathematical model for dengue fever transmission that incorporates vector control, treatment, and mass awareness; their results, supported by stability, bifurcation, and sensitivity analyses, highlight that simultaneous application of these interventions is substantially more effective in reducing disease spread than using them individually or in pairs. Ndendya et al. (2025) \cite{Ndendya2025b} developed a comprehensive model demonstrating how vaccination, treatment, and public health education can significantly reduce norovirus transmission. Similaryly, Liana et al. (2025) \cite{LIANA2025e02516} developed a deterministic model that integrates tuberculosis transmission with nutritional status, revealing through bifurcation and stability analysis that malnutrition not only increases TB susceptibility but also amplifies disease burden—emphasizing the need for nutritional interventions alongside conventional TB control strategies. Recently, a fractional-order rabies transmission model using the Atangana–Baleanu–Caputo derivative, combined with pre- and post-exposure prophylaxis and MCMC parameter estimation, demonstrated that memory effects crucially influence disease dynamics and vaccination outcomes \cite{NDENDYA2025e02800}. While public health education has been shown to significantly reduce the spread of conjunctivitis in a recent model developed in \cite{Ndendya2025}, antibiotic resistance remains a persistent challenge in bacterial infections. To address the issue of antibiotic resistance, many works have been done several times. Daşbaşı and Öztürk (2016) \cite{Daşbaşı2016} have discussed a mathematical model that evaluates the effectiveness of antibiotic treatments in the presence of bacterial resistance and immune system response. According to a study conducted by Khan and Imran
(2018) \cite{Adnan2018} chemostat-based model, the concentration of antibiotics changes over time according to the dosage strategy, and nutrients enter and exit the system. After that, Gjini et al. (2020)  \cite{Gjini2020} developed a control strategy model to demonstrate how antibiotic treatment effectiveness varies with timing, duration, and kill rate infections and suggested that a more empirical understanding of bacterial infection processes in individual hosts is necessary for antibiotic stewardship advancements. Also, the interaction between antibiotic therapy and the immune system response is important in reducing the bacterial population proposed through a system of ordinary differential equations by  Mondragon et al. (2021) \cite{Mondragon2021}. On this track, Nashebi et al. (2024) \cite{Nashebi2024} proposed a mathematical model whose analytical and numerical results reveal that antagonistic interactions against wildtype bacteria reduce resistance development more than those against mutants, emphasizing the importance of considering antibiotic interactions against wildtype bacteria rather than mutants. Again, research conducted by Khurana et al. (2024)  \cite{Khurana2024} found that empiric narrow-spectrum antibiotic usage lowers second- and third-line antibiotic resistance but increases resistance and death as well as treatment duration and mortality need to be reduced to convert to narrow-spectrum antibiotics. 
Also, a recent study in north India illustrated a significant presence of multidrug-resistant (MDR) for uptaking antibiotic and extended-spectrum beta-lactamase (ESBL) producing pathogens in livestock, poultry, and retail meat, with atypical EPEC being the most prevalent among diarrheal cases. Additionally, the study identified C. hyointestinalis, an emerging zoonotic pathogen, for the first time in India \cite{MAHINDROO2024}. Addressing antibiotic resistance as a global health crisis with its key ARB/ARG reservoirs, Karnwal et al. (2025) \cite{Karnwal2025} explored its drivers, challenges, and promising interventions—highlighting the need for integrated strategies including dietary modulation, novel therapeutics, and strengthened regulatory frameworks.

Upon thorough review of the available literature, it has been revealed that no research currently exists to identify an ideal strategy for achieving the highest possible number of immune cells while minimizing the expense of administering antibiotics. Meanwhile, there is a lack of research that shows how antibiotic over or under-dosing affects the immune system of humans described by a system of ordinary differential equations incorporating a new function. In order to figure out a suitable strategy for maximizing immune cells, this research proposes a mathematical model with an objective function. Additionally, the consequences of both overdosing and underdosing from recommended antibiotic doses have been found by introducing a floor function in modeling. 

The following portion of the article is organized in the following manner: The model for the action of antibiotics has been established in Section 2. Section 3 has presented several analytical analyses of the model in distinct subsections. After this, Section 4 presents the model's numerical simulation along with the results of under and over-dosing antibiotic treatments by way of introducing a floor function. Section 5 thereafter has provided the results together with associated observations. The last part, namely Section 6, highlights the conclusions obtained from this study.

 \section{ Design of an Antibiotic's Effect Model}
Antibiotic treatment is one of the most effective ways of treatment in case of bacterial disease. But at the same time, this effective formula also causes a great problem in the human health care sector by making the bacteria resistant against these antibiotics. Careless use of antibiotics accelerates this mutation. So we are in great danger because if we do not control this mutation, scientists are warning that we could soon return to the “dark age of medicine”, where all of our available drugs become ineffective against these resistant bacteria.\\
In this part of the article, we are going to develop a mathematical model using the concept of optimization based on the previous model mentioned in Daşbaşı and Öztürk (2016) \cite{Daşbaşı2016} to describe the dynamics of host antibiotic resistance of bacteria and the corresponding immune response in the presence of a single antibiotic. Although the present model assumes that antibiotics act immediately after administration, there may be a time delay before the drug becomes effective. To reduce model complexity and maintain analytical tractability, this delay effect would be neglected in the current study. Based on the biological phenomena described in several studies \cite{Mondragón2014, Daşbaşı2016, Esteva2021, Coll2009, Pugliese2008}, we consider the following four variables:
$S(t)$= Population size of Sensitive Bacteria at time t\\
$R(t)$= Population size of Resistance Bacteria at time t\\
$I(t)$ = Population size of Immune cells at time t\\
$A(t)$ = Antibiotic concentration at time t \\
The bacteria in the human body are divided in our body into antibiotic-sensitive and antibiotic-resistant bacteria that follow a logistic growth with carrying capacity C. Let $\beta_s$ and $\beta_r$  be the birth rates of sensitive and resistant bacteria respectively. The rate of change of sensitive and resistant bacteria depends on several parameters. The number of sensitive bacteria increases due to the birth rate of sensitive bacteria $\beta_s$ until the number of bacteria exceeds the carrying capacity. This is denoted by the term $\beta_s S\left( 1-\frac{S+R}{C}\right)$. The number of sensitive bacteria decreases due to the immune response of the 
host body, which is $\lambda SI$, where $\lambda$ is the per capita death rate of sensitive bacteria. Sensitive bacteria also decrease as a result of antibiotic concentration in the host body, which is $\Bar{\sigma}SA$, where $\Bar{\sigma}$ is the death rate of sensitive bacteria due to the exposure to given antibiotics. Finally, sensitive bacteria decrease as a result of the mutation of sensitive bacteria to resistant bacteria due to exposure to the given antibiotics, which are given by the term $\Bar{\alpha}SA$, where $\Bar{\alpha}$ is the mutation rate. Hence, the equation will be,
\begin{align}
    \frac{dS}{dt}=\beta_sS\left(1-\frac{S+R}{C}\right)-\lambda SI-\Bar{\alpha}SA-\Bar{\sigma}SA
\end{align}
Again, the number of resistant bacteria increases due to the birth rate of resistant bacteria $\beta_r$ until the number of bacteria has exceeded the carrying capacity. So, this is denoted by the term  $\beta_rR\left( 1-\frac{S+R}{C}\right)$.\\
\begin{figure}[hbt!]
   \centering
    \includegraphics[width=0.45\textwidth]{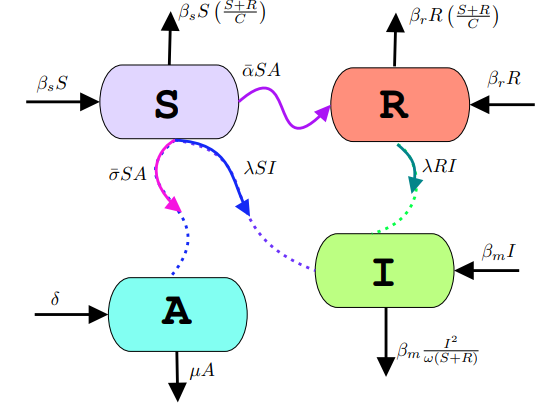}
    \caption{Schematic diagram of the model}
    \label{fig:model diagram}
\end{figure}
In addition, the rise of resistant bacteria is also attributed to the mutation of sensitive bacteria into resistant bacteria because of exposure to the antibiotic $\Bar{\alpha}SA$. Conversely, the number of resistant bacteria declines due to the immune response of the host body, denoted as $\lambda RI$, where $\lambda$ represents the per capita death rate of resistant bacteria. The equation can be expressed as follows:
\begin{align}
    \frac{dR}{dt}=\beta_r R\left(1-\frac{S+R}{C}\right)-\lambda RI+\Bar{\alpha}SA
\end{align}
Now, using a logistic term, we can see that immune cells are being brought to the infection region at a rate $\beta_m$, and their carrying capacity is equal to $\omega$ times the number of bacteria existing there. Concerning the expansion of immune cells, this is of great biological importance. These interactions may also be represented as a generic model of a bacterial function, as that of \textit{Mycobacterium tuberculosis}. Based on this information, the equation will be
\begin{align}
    \frac{dI}{dt}=\beta_m I\left(1-\frac{I}{\omega(S+R)}\right)
\end{align}
Once again, assuming that $\delta$ represents the constant rate of antibiotic supply and $\mu$ represents the per capita rate of antibiotic uptake, the rate of change in antibiotic concentration inside the host body will be,
\begin{align}
    \frac{dA}{dt}=\delta-\mu A
\end{align}
Hence, our designed model equations involving a system of ordinary differential equations as like \cite{khananik2020} which is a particular case of the model formulated and analzsed by Daşbaşı and Öztürk (2016) in \cite{Daşbaşı2016},
\begin{align}
     \frac{dS}{dt} &=\beta_sS\left(1-\frac{S+R}{C}\right)-\lambda SI-\Bar{\alpha}SA-\Bar{\sigma}SA \\
       \frac{dR}{dt}&=\beta_r R\left(1-\frac{S+R}{C}\right)-\lambda RI+\Bar{\alpha}SA \\
       \frac{dI}{dt} &=\beta_m I\left(1-\frac{I}{\omega(S+R)}\right) \\
        \frac{dA}{dt}& =\delta-\mu A 
\end{align}
with $S(0)=S_0\ge0, R(0)=R_0\ge0, I(0)=I_0\ge0$ and $A(0)=A_0\ge0$.\\
In order to decrease the quantity of parameters or make the dimensionless of the above model, the variables are modified in the following manner:
\begin{equation}
    s=\frac{S}{C}, r=\frac{R}{C}, m=\frac{I}{\omega C}, a=\frac{A}{\frac{\delta}{\mu}}
\end{equation}
Therefore, the addition of these variables transforms our model containing a system of ordinary differential equations as the following form:
\begin{align}
    \frac{ds}{dt}&=\beta_s s\left(1-(s+r)\right)-\eta sm-\alpha sa-\sigma s a \label{equ:model equation1}\\
    \frac{dr}{dt}&=\beta_r r(1-(s+r))-\eta rm+\alpha sa \label{equ:model equation2}\\
    \frac{dm}{dt}&=\beta_m m\left(1-\frac{m}{(s+r)}\right)\label{equ:model equation3}\\
    \frac{da}{dt}&=\mu(1-a)\label{equ:model equation4}
\end{align}
where 
\begin{equation*}
    \alpha=\Bar{\alpha}\left(\frac{\delta}{\mu}\right),\sigma=\Bar{\sigma}\left(\frac{\delta}{\mu}\right), \eta=\lambda\omega C
\end{equation*}
Also, this model is represented as a diagram in Figure \ref{fig:model diagram}.\\
Furthermore, antibiotics, when administered, selectively and efficiently eradicate the sensitive bacteria responsible for the disease.  At the same time, antibiotics provide a suitable preference that may cause certain bacteria to develop resistance. These bacteria may have altered genes or other ways to avoid the drug's effects. Over time, the immune system becomes less effective against certain bacteria that have developed resistance. As a result, as much as we would raise the strength of antibiotics for uptaking, the cost of buying antibiotics as well as resistant bacteria would also increase. So, it is high time to determine the situation for finding the maximum level of immune cells while simultaneously reducing the presence of resistant bacteria and antibiotic concentration while maintaining the cost. Therefore, our objective equation is 
\begin{equation}
 \textbf{Maximize} \int_{0}^{t} \left[B_1 m(t)-B_2 r(t)-B_3 a(t)\right] dt \label{eq:Integration}
\end{equation}
Subjected to 
\begin{align}
    \frac{ds}{dt}&=\beta_s s\left(1-(s+r)\right)-\eta sm-\alpha sa-\sigma s a \label{equ:model equation12}\\
    \frac{dr}{dt}&=\beta_r r(1-(s+r))-\eta rm+\alpha sa \label{equ:model equation22}\\
    \frac{dm}{dt}&=\beta_m m\left(1-\frac{m}{(s+r)}\right)\label{equ:model equation32}\\
    \frac{da}{dt}&=\mu(1-a)\label{equ:model equation42}
\end{align}
with the initial conditions $s(0)=s_0\ge 0, r(0)=r_0\ge0, m(0)=m_0\ge 0$ and $a(0)=a_0\ge 0$ and the parameters \( B_1 \), \( B_2 \), and \( B_3 \) are weight factors in the objective functional, which is designed to maximize the concentration of immune cells while penalizing resistant bacteria and limiting antibiotic usage. These weights reflect the relative importance or priority given to each of these goals. A higher value of a particular \( B_i \) (for \( i = 1, 2, 3 \)) indicates a greater emphasis on that component during the maximization process.

\section{Analytical Analysis of Model}
\subsection{Positivity of the Model}
\begin{theorem}
   \label{Theorem:Positivity Theorem}\index{Positivity}
Considering $s(0)\ge 0, r(0)\ge 0, m(0)\ge 0$ and $a(0)\ge0$, then $s(t),r(t),m(t)$ and $a(t)$ will be always positive for all $t \epsilon [0,T]$ in $R^{+}_4$ where $T>0$. 
\end{theorem}
\begin{proof}
    Taking all parameters of the system and all initial values to be positive, we have to prove that $s(t), r(t),m(t)$ and $a(t)$ will be positive for all  $t \epsilon [0,T]$ in $R^{+}_4$.\\
From the equation (\ref{equ:model equation3}), we can write  as follows
\begin{align}
    \frac{dm}{dt}&=\beta_m m\left(1-\frac{m}{s+r}\right) \\
    \Longrightarrow \frac{dm}{m}&=\beta_m \left(1-\frac{m}{s+r}\right) dt 
\end{align}
Now integrating both sides, we get
\begin{align}
 ln(m) &= \beta_m t-\beta_m t\int\frac{m}{s+r}dt \\
 \Longrightarrow m &= e^{\beta_m t-\beta_m t\int\frac{m}{s+r}dt} > 0
\end{align}
Again, equation (\ref{equ:model equation4})
shows,
\begin{align}
    \frac{da}{dt}&=\mu-\mu a \\
    \Longrightarrow \frac{da}{dt} & \ge -\mu a \\
    \Longrightarrow \frac{da}{a} & \ge -\mu dt \\
    \Longrightarrow ln(a) & \ge -\mu t
\end{align}
After including initial conditions,
\begin{align*}
    \therefore a(t) & \ge a_0e^{-\mu t} >0
\end{align*}
We have already got $m(t)>0$ and $a(t)>0$, so consider that $\eta m=\eta_0$ and $\alpha a=\alpha_0$ are positive always. Then from equation \ref{equ:model equation1}, we get, 
\begin{align}
    \frac{ds}{dt}&=\beta_s s(1-(s+r))-\eta_0 s-\alpha_0s-\sigma_0s \\
 \Longrightarrow   \frac{ds}{s}&=\beta_s (1-(s+r))-\eta_0 -\alpha_0-\sigma_0 \\
 \Longrightarrow ln s &=(\beta_s-\eta_0-\alpha_0-\sigma_0)t-\beta_s \int{(s+r)dt} \\
 \therefore s(t) &=e^{(\beta_s-\eta_0-\alpha_0-\sigma_0)t-\beta_s \int{(s+r)dt}}>0
\end{align}
Again, From equation \ref{equ:model equation2}, we get
\begin{align}
    \frac{dr}{dt}& \ge \beta_r r(1-(s+r))-\eta_0 r \\
    \Longrightarrow  \frac{dr}{r}& \ge \{\beta_r (1-(s+r))-\eta_0\} dt \\
    \therefore r(t) & \ge e^{\{\beta_rt -\int{(s+r)dt}-\eta_0t\}} >0
\end{align}
This completes our proof.
\end{proof}

\subsection{Feasible Region of the Model}
\begin{theorem}
    The model is biologically feasible in the region
    \begin{align}
        \Lambda &=\left\{(s,r,m,a)\epsilon R_+^4:0\le s,r,0\le m\le s+r\le 1,0\le a\le 1\right\}
    \end{align}
\end{theorem}
\begin{proof}
    Adding first two equations of our model, we get
    \begin{align}
        \frac{ds}{dt}+\frac{dr}{dt} &= (\beta_s s+\beta_r r)(1-(s+r))-\eta m(s+r)-\sigma s a
    \end{align}
can be observed. Taking the region $\Lambda$, we get the following inequality:
\begin{align}
    \frac{d(s+r)}{dt}\le \beta_s (s+r)(1-(s+r))
    \label{feasible}
\end{align}
By the solution according to $(s+r)$ of inequality \eqref{feasible}, it follows that $0\le s+r\le 1$ for all $t\ge 0$. Furthermore, the solution of the last equation of the system (\ref{equ:model equation1}-\ref{equ:model equation4}) are 
\begin{align}
    a(t)=1+(-1+a(0))e^{-\mu t}
\end{align}
where the initial conditions satisfy $0\le a(0)\le 1$.\\
At the end, we consider that $0\le s+r=\nu(constant)\le 1$. Then the solution for the third equation in our considering model will be
\begin{align}
    m=\frac{\nu}{1+e^{-kt-m(0)}}
      \label{feasible2}
\end{align}
where initial conditions satisfy $0<m(0)\le s(0)+r(0)$. From the equation (\ref{feasible2}), it is found that $0\le m\le s+r\le 1$.\\
Thus we can say that all the solutions of our model are bounded within the region $\Lambda$.

\end{proof}

\subsection{Existence and Uniqueness of Model's Solution}
\begin{Lemma}
   Consider the domain $S$ that fulfils the Lipschitz condition. The system's solutions are both existent and unique for any time $T\ge0$ within the domain $D$, provided any non-negative initial conditions.
\end{Lemma}
\begin{proof}
This theorem from Sowole et al. (2019) \cite{sowole2009} shows that the suggested Lipschitz conditions must be followed to make sure that there is and only exists one solution in a certain region, denoted as $S$. Let's take
\begin{align}
   f(s,r,m,a)= \frac{ds}{dt}&=\beta_s s\left(1-(s+r)\right)-\eta sm-\alpha sa-\sigma s a \label{equ:un equation1}\\
  g(s,r,m,a)=   \frac{dr}{dt}&=\beta_r r(1-(s+r))-\eta rm+\alpha sa \label{equ:un equation2}\\
   h(s,r,m,a)=  \frac{dm}{dt}&=\beta_m m\left(1-\frac{m}{(s+r)}\right)\label{equ:un equation3}\\
  p(s,r,m,a)=   \frac{da}{dt}&=\mu(1-a)\label{equ:un equation4}
\end{align}

The following is the partial derivatives of $f,g,h,p$ concerning variables $s,r,m,a$, as calculated using the system's equation from above:
 \begin{equation*}
	\frac{\partial f}{\partial s}=(1-2s-r)\beta_s-\eta m -(\alpha +\sigma) a
	\hspace{0.5 cm} 
	\bigg\lvert\frac{\partial f}{\partial s}\bigg\rvert= \bigg\lvert (1-2s-r)\beta_s-\eta m -(\alpha +\sigma) a \bigg\rvert <\infty
\end{equation*}

\begin{equation*}
	\frac{\partial f}{\partial r}=-\beta_s s^2
	\hspace{1.2 cm} 
	\bigg\lvert\frac{\partial f}{\partial r}\bigg\rvert= \bigg\lvert -\beta_s s^2 \bigg\rvert=\beta_s s^2 <\infty
\end{equation*}
\begin{equation*}
	\frac{\partial f}{\partial m}=-\eta s
	\hspace{1.2 cm} 
	\bigg\lvert\frac{\partial f}{\partial m}\bigg\rvert= \bigg\lvert -\eta s\bigg\rvert= \eta s <\infty
\end{equation*}

\begin{equation*}
	\frac{\partial f}{\partial a}= -(\alpha+\sigma)s
	\hspace{1.2 cm} 
	\bigg\lvert\frac{\partial f}{\partial a}\bigg\rvert= \bigg\lvert -(\alpha+\sigma)s\bigg\rvert= (\alpha+\sigma)s <\infty
\end{equation*}

Again,
\begin{equation*}
	\frac{\partial g}{\partial s}=\alpha a-\beta_r r
	\hspace{0.9 cm} 
	\bigg\lvert\frac{\partial f}{\partial s}\bigg\rvert= \bigg\lvert \alpha a-\beta_r r \bigg\rvert <\infty
\end{equation*}

\begin{equation*}
	\frac{\partial g}{\partial r}=(1-s-2r)\beta_r-\eta m
	\hspace{0.6 cm} 
	\bigg\lvert\frac{\partial f}{\partial r}\bigg\rvert= \bigg\lvert (1-s-2r)\beta_r-\eta m \bigg\rvert <\infty
\end{equation*}
\begin{equation*}
	\frac{\partial g}{\partial m}=-\eta r
	\hspace{1.2 cm} 
	\bigg\lvert\frac{\partial f}{\partial m}\bigg\rvert= \bigg\lvert -\eta r\bigg\rvert= \eta r <\infty
\end{equation*}

\begin{equation*}
	\frac{\partial g}{\partial a}= \alpha s
	\hspace{1.2 cm} 
	\bigg\lvert\frac{\partial f}{\partial a}\bigg\rvert= \bigg\lvert \alpha s \bigg\rvert= \alpha s <\infty
\end{equation*}

In a similar fashion we can readily demonstrate that,
\begin{align}
& \abs{\frac{\partial h}{\partial s}},\abs{\frac{\partial h}{\partial r}},\abs{\frac{\partial h}{\partial m}},\abs{\frac{\partial h}{\partial a}}<\infty \\
& \abs{\frac{\partial p}{\partial s}},\abs{\frac{\partial p}{\partial r}},\abs{\frac{\partial p}{\partial m}},\abs{\frac{\partial p}{\partial a}}<\infty
\end{align}
Therefore, it has been shown that every partial derivative exhibits both continuity and boundedness. This confirmation verifies that the requirements specified by the Lipschitz criterion have been met. Consequently, based on the ideas outlined in
Sowole et al. (2019) and Kundu and Mallick (2023) \cite{sowole2009,Kundu2023}, there is a unique solution for the system (\ref{equ:model equation1}-\ref{equ:model equation4}) only inside the specified region $S$. This assumption is substantiated by the evidence provided.
\end{proof}
Therefore, the existence, positivity and boundedness of the solutions of the proposed
model (\ref{equ:model equation1}-\ref{equ:model equation4}) ensures that the model has a mathematical and biological meaning.
\subsection{Equilibrium Point}
To find the equilibrium point of the proposed model, we solve the system of equations derived from the model's structure. Our equilibrium point is obtained by setting the derivatives to zero as follows:
\begin{align*}
& \beta_s s\left(1-(s+r)\right)-\eta sm-\alpha sa-\sigma s a =0\\ 
& \beta_r r(1-(s+r))-\eta rm+\alpha sa  =0\\
& \beta_m m\left(1-\frac{m}{(s+r)}\right) =0\\ 
&\mu(1-a)=0
\end{align*}
Thus, after solving the algebraic equations above, we obtained the five non-negative equilibrium points. They are
\begin{align*}
    E_1 &=\left(0,0,0,1\right) \\
     E_2 &=\left(0,1,0,1\right)\\
      E_3 &=\left(\frac{\beta_s(\alpha b_1-b)+(\alpha+\sigma)b}{\beta_s b},\frac{\alpha b_1}{b},0,1\right) \\
       E_4 &=\left(0,\frac{\beta_r}{\beta_r+\eta},\frac{\beta_r}{\beta_r+\eta},1\right) \\
        E_5 &=\left(\frac{-b_1\left(\beta_r\alpha+\beta_r\sigma+\eta\psi\right)}{(\beta_s+\eta)(b-\tau)},\frac{-b_1\alpha}{-b+\eta\tau},\frac{-b_1}{\beta_s+\eta},1\right)
\end{align*}
where $b=\beta_r(\alpha+\sigma)-\beta_s\alpha$, $b_1=\alpha+\sigma-\beta_s$, $\tau=\beta_s-\beta_r-\sigma$ and $\psi=\alpha-\tau$ with $\alpha<\beta_s-\sigma$, $\beta_r(\alpha+\sigma)<\beta_s\alpha$  and $\beta_s\alpha b_1>b$.
\subsection{Status of Resistant Bacteria}
Through the use of the next generation matrix approach highlighted in Diekmann et al. (1990) \cite{Diekmann1990OnTD}, we would be able to figure out the upcoming scenario of resistant bacteria. For the purpose of discussion, the state variable $r$ has been taken into consideration in order to investigate the dynamical behavior of bacteria that are resistant to antibiotics and pose a risk to the immune system of a human being. Now, equation (\ref{equ:model equation2}) tells
\begin{equation}
   \frac{dr}{dt}=\beta_r r[1-(s+r)]-\eta rm+\alpha sa
\end{equation}
Differentiating it with respect to $r$, we get
\begin{equation}
    \Longrightarrow \frac{d}{dr}\left(\frac{dr}{dt}\right)_{s=s^*,r=r^*,m=m^*,a=a^*}=\beta_r -(\beta_r s^*-2\beta_rr^*-\eta m^*)
\end{equation}
Therefore, two matrices $M$ and $D$ which represent the boosting and falling number of resistant bacteria in body are $M=\beta_r$ and $D=\beta_r s^*-2\beta_rr^*-\eta m^*$.\\
Then the future situation of guava borer would be $R_0$ will be
\begin{equation}
    R_0=\frac{M}{D}=\frac{\beta_r}{\beta_r s^*-2\beta_rr^*-\eta m^*}
\end{equation}
Thus, the immune system  will be improved by declining resistant bacteria if  $\beta_r< \beta_r s^*-2\beta_rr^*-\eta m^*$. But, the immune system gradually worsens by rising resistant bacteria if $\beta_r> \beta_r s^*-2\beta_rr^*-\eta m^*$.

\subsection{Stability Analysis}
To find the stability of our proposed model (\ref{equ:model equation1})-(\ref{equ:model equation4}) at the equilibrium point $(s^*,r^*,m^*,a^*)$, firstly we consider
\begin{align}
& f_1(s^*,r^*,m^*,a^*)=\beta_s s^*\left(1-(s^*+r^*)\right)-\eta s^*m^*-\alpha s^*a^*-\sigma s^* a^* =0
\label{Eq:s1}
\\ 
& f_2(s^*,r^*,m^*,a^*)=\beta_r r^*(1-(s^*+r^*))-\eta r^*m^*+\alpha s^*a^*  =0
\label{Eq:s2}
\\
& f_3(s^*,r^*,m^*,a^*)=\beta_m m^*\left(1-\frac{m^*}{(s^*+r^*)}\right) =0
\label{Eq:s3}
\\ 
& f_4(s^*,r^*,m^*,a^*)=\mu(1-a6*)=0
\label{Eq:s4}
\end{align}
For the equation (\ref{Eq:s1}) -- (\ref{Eq:s4}), the Jacobian matrix is 
\begin{equation}
J = \frac{\partial(f_1, f_2, f_3, f_4)}{\partial(s^*, r^*, m^*, a^*)} = 
\begin{pmatrix}
\frac{\partial f_1}{\partial s^*} & \frac{\partial f_1}{\partial r^*} & \frac{\partial f_1}{\partial m^*} & \frac{\partial f_1}{\partial a^*} \\
\frac{\partial f_2}{\partial s^*} & \frac{\partial f_2}{\partial r^*} & \frac{\partial f_2}{\partial m^*} & \frac{\partial f_2}{\partial a^*} \\
\frac{\partial f_3}{\partial s^*} & \frac{\partial f_3}{\partial r^*} & \frac{\partial f_3}{\partial m^*} & \frac{\partial f_3}{\partial a^*} \\
\frac{\partial f_4}{\partial s^*} & \frac{\partial f_4}{\partial r^*} & \frac{\partial f_4}{\partial m^*} & \frac{\partial f_4}{\partial a^*}
\end{pmatrix}
\label{Eq:jaco}
\end{equation}
Thus, the Jacobian matrix $J$ for the proposed model at the equilibrium point using equation \eqref{Eq:jaco} is given by\\
$J_{\arrowvert s^*, r^*, m^*, a^*}=$
\begin{equation}
\resizebox{\textwidth}{!}{$
\begin{bmatrix}
	\beta_s - 2\beta_s s^* - \beta_s r^* - \eta m^* - \alpha a^* - \sigma a^*  & \beta_s s^* & -\eta s^*  & -s^*\alpha - s^*\sigma \\
	\beta_r s^* + \alpha a^* & \beta_r - \beta_r s^* - 2\beta_r r^* - \eta m^* & -\eta r^* & s^* \alpha \\
	\frac{\beta_m m^{*2}}{(s^* + r^*)^2} & \frac{\beta_m m^{*2}}{(s^* + r^*)^2}  & \beta_m - \frac{2\beta_m m^*}{s^* + r^*}  & 0\\
	0 & 0 & 0 & -\mu
\end{bmatrix}
$}
\end{equation}
Now we will examine whether the equilibrium points of the model $E_1,E_2,E_3,E_4,E_5$ are stable or not.\\

1. At the equilibrium point $E_1=(0,0,0,1)$ Jacobian matrix is not defined. So the equilibrium point $E_1$ is unstable.\\

2. At the equilibrium point $E_2=(0,1,0,1)$, Jacobian matrix $J_{\arrowvert E_2}$ is,
\begin{equation}
J_{\arrowvert E_2}= \begin{bmatrix}
	-\alpha-\sigma & 0 & 0  & 0 \\
	\alpha-\beta_r & -\beta_r  & -\eta & 0 \\
 0 & 0 &\beta_m & 0\\
 0& 0 & 0 & -\mu
\end{bmatrix}
\end{equation}
Let $\lambda$ be the eigenvalue and $I$ be the identity matrix, then the characteristic equation is,
\begin{equation}
   J_{\arrowvert E_2}= \begin{bmatrix}
	-\alpha-\sigma-\lambda & 0 & 0  & 0 \\
	\alpha-\beta_r & -\beta_r-\lambda & -\eta & 0 \\
 0 & 0 &\beta_m-\lambda & 0\\
 0& 0 & 0 & -\mu-\lambda
\end{bmatrix}=0
\end{equation}
By solving this characteristic equation we get, 
\begin{align*}
   & (-\alpha-\sigma-\delta)(-\beta_r-\lambda)(\beta_m-\lambda)(-\mu-\lambda)=0 \\
   & \Longrightarrow \lambda_1=-\mu, \lambda_2=\beta_m, \lambda_3=-\beta_r, \lambda_4=-\alpha-\sigma 
\end{align*}

Since the value of $\lambda_2$ is $\beta_m$ which is positive, therefore the equilibrium point $(0,1,0,1)$ is unstable.\\

In the same way, applying the idea of Routh Hurwitz criteria, the equilibrium point $E_3$ is unstable but the equilibrium point $E_4$ is asymptotically stable if $\alpha>\beta_s$, whereas $E_5$ is asymptotically stable if $\alpha>\beta_s>\sigma$ and $\beta_r>\beta_s>n$.

\subsection{Characteristics of States Equilibrium points} 
From the equation (\ref{equ:model equation1})-(\ref{equ:model equation4}) of the model, we get a couple of equations of the characteristic function as,

\begin{align}
    f(s^*,r^*,\eta)&=\beta_s s^*[1-(s^*+r^*]-s^*\alpha -s^*\sigma-\eta s^*(s^*+r^*) \\
     g(s^*,r^*,\eta)&=\beta_r r^*[1-(s^*+r^*]-s^*\alpha -\eta r^*(s^*+r^*) 
\end{align}
Now rate of change of sensitive bacteria $(s)$ with respect to immune response $\eta$ is given by,
\begin{align}
    \therefore \frac{ds^*}{d\eta}&=\frac{\begin{vmatrix}
\frac{\partial{ f(s^*, r^*,\eta)}}{\partial{r^*}} & \frac{\partial{  f(s^*, r^*,\eta)}}{\partial{\eta}} \\
\frac{\partial{ g(s^*, r^*,\eta)}}{\partial{r^*}} & \frac{\partial{ g(s^*, r^*,\eta)}}{\partial{\eta}}
\end{vmatrix}}
{\begin{vmatrix}
\frac{\partial{ f(s^*, r^*,\eta)}}{\partial{s^*}} & \frac{\partial{ f(s^*, r^*,\eta)}}{\partial{r^*}} \\
\frac{\partial{  g(s^*, r^*,\eta)}}{\partial{s^*}} & \frac{\partial{ g(s^*, r^*,\eta)}}{\partial{r^*}}
\end{vmatrix}}\\
&=\frac{\frac{\partial f}{\partial r^*}\frac{\partial g}{\partial \eta}-\frac{\partial f}{\partial \eta}\frac{\partial g}{\partial r^*}}{\frac{\partial f}{\partial s^*}\frac{\partial g}{\partial r^*}-\frac{\partial g}{\partial s^*}\frac{\partial f}{\partial r}} \\
& = \frac{d_2 s^{*2}r^{*3}+s{*2}+s^*r^*}{[-c_1-2d_1+d_2r^*][\beta_r-2d_3r^*-d_3s^*]-[\alpha-d_3r^*]d_2s^*} 
\end{align}
where $d_1=\beta_s+\eta,d_2=\beta_s-\eta, d_3=\beta_r+\eta ,c_1=\alpha+\sigma-\beta_s$ 
It is seen that the numerator is always positive but the denominator is negative if $d_2r^*<c_12d_1, \beta_r>d_3s^*+2d_1r^*$ with $\alpha>d_1 r^*$.
For this reason, we conclude that $\frac{ds^*}{d\eta}<0$.
This implies that the derivative of sensitive bacteria with regard to immune response is negative, signifying that an increase in immune response will result in a reduction in the number of sensitive bacteria. Likewise, a reduction in immune response will result in a boost in the number of susceptible bacteria.

Again, by using the relationship between the immune response $\eta$ and the rate of change of resistant bacteria (r), we get

\begin{align}
    \therefore \frac{dr^*}{d\eta}&=\frac{\begin{vmatrix}
\frac{\partial{ f(s^*, r^*,\eta)}}{\partial{\eta^*}} & \frac{\partial{  f(s^*, r^*,\eta)}}{\partial{s^*}} \\
\frac{\partial{ g(s^*, r^*,\eta)}}{\partial{\eta}} & \frac{\partial{ g(s^*, r^*,s)}}{\partial{s^*}}
\end{vmatrix}}
{\begin{vmatrix}
\frac{\partial{ f(s^*, r^*,\eta)}}{\partial{s^*}} & \frac{\partial{ f(s^*, r^*,\eta)}}{\partial{r^*}} \\
\frac{\partial{  g(s^*, r^*,\eta)}}{\partial{s^*}} & \frac{\partial{ g(s^*, r^*,\eta)}}{\partial{r^*}}
\end{vmatrix}}\\
&=\frac{\frac{\partial f}{\partial \eta^*}\frac{\partial g}{\partial s^*}-\frac{\partial f}{\partial s^*}\frac{\partial g}{\partial \eta}}{\frac{\partial f}{\partial s^*}\frac{\partial g}{\partial r^*}-\frac{\partial g}{\partial s^*}\frac{\partial f}{\partial r^*}} \\
& = \frac{(-s^*r^*-s^{*2})(\alpha-d_3r^*)-(-c_-2d_1+d_2r^*)(-r^{*2}-r^*s^*) }{[-c_1-2d_1+d_2r^*][\beta_r-2d_3r^*-d_1s^*]-[\alpha-d_3r^*]d_2s^*} 
\end{align}

Using the previous conditions, $\frac{dr^*}{d\eta}<0$. This points to a negative relationship between the rate of change of resistant bacteria (r) and the immune response. That is, the number of bacteria that are resistant (r) will decrease as the immune response increases. Similarly, the number of bacteria that are resistant to antibiotics (r) will rise if the immune response decreases.\\ In the similar way, it has been also obtained that $\frac{dm^*}{d\eta}>0$ and $\frac{da^*}{d\eta}=0$.\\
However, a summary of the result for characteristics of the equilibrium point with a parameter $\eta$ is presented in Table \ref{tab:summary-characteristics}.
\begin{table}[hbt!]
\centering
\caption{Summary of the result for characteristics of the equilibrium points regarding the parameter}
\begin{tabular}{c c c c}
\hline
\textbf{Parameter} & \textbf{Results} & \textbf{Conditions} & \textbf{Interpretations} \\
\hline
\multirow{4}{*}{$\eta$} 
& $\frac{ds^*}{d\eta} < 0$ 
& \multirow{4}{*}{\begin{tabular}[c]{@{}l@{}}$(\beta_s - \eta)r^* < 2(\alpha + \sigma - \beta_s)$\\
$(\beta_s + \eta)$,\\
$\beta_r > (\beta_r + \eta)s^* + 2(\beta_s + \eta)r^*$\\
$\alpha > (\beta_s + \eta)r^*$
\end{tabular}} 
& \begin{tabular}[c]{@{}l@{}}Increasing the value of \( \eta \) leads to a\\ decrease in the population of\\ sensitive bacteria.\end{tabular} \\
\cline{2-2} \cline{4-4}
& $\frac{dr^*}{d\eta} < 0$ 
&  
& \begin{tabular}[c]{@{}l@{}}A higher bacteria eradication rate\\ by immune cells (\( \eta \)) contributes to\\ a reduction in the resistant\\ bacteria population.\end{tabular} \\
\cline{2-2} \cline{4-4}
& $\frac{dm^*}{d\eta} > 0$ 
&  
& \begin{tabular}[c]{@{}l@{}}The immune cell population incre-\\ases with increasing \( \eta \).\end{tabular} \\
\cline{2-2} \cline{4-4}
& $\frac{da^*}{d\eta} = 0$ 
&  
& \begin{tabular}[c]{@{}l@{}}The change of antibiotic concen-\\ tration remains constant with\\ respect to \( \eta \).\end{tabular} \\
\hline
\end{tabular}
\label{tab:summary-characteristics}
\end{table}
\subsection{Dynamics of Resistant Bacteria under Declining Immune Cells and Monotonically Increasing Antibiotic Concentration}
In this part, we will investigate the dynamics of resistant bacteria in light of the gradually diminishing immune cells and the steadily growing concentration of antibiotics. The fact that antibiotic concentrations are continuing to climb as a result of taking medicines indicates that the levels that are currently present must be higher in comparison to those that were present in the past.\\
With regard to time t, we differentiate both sides of the equation (\ref{equ:model equation2}),
\begin{align}
	\frac{dr}{dt} &=\beta_r r[1-(s+r)]-\eta rm+\alpha sa\\
	\Longrightarrow \frac{d^2r}{dt^2}   & 
	> \beta_r \frac{dr}{dt}(-1-s)-\beta_r r\frac{ds}{dt}-\eta m \frac{dr}{dt}+\alpha a \frac{ds}{dt}
\end{align}
Because of the fact that immune cells are decreasing with rising the concentration of antibiotics in body, so $\frac{dm}{dt}$ must be negative and so $\alpha s \frac{da}{dt}$  and $-\eta r \frac{dm}{dt}$ must be positive term.\\ 
Again for the critical points or equilibrium points,
\begin{equation}
	\frac{ds}{dt}=\frac{dr}{dt}=\frac{dm}{dt}=\frac{da}{dt}=0
\end{equation}
Then we find that, at the equilibrium point
\begin{equation}
\frac{d^2r}{dt^2}>0
\end{equation}
That means at the critical points, the resistant bacteria is minimal.
Biologically, this result shows that the concentration of bacteria that are resistant to antibiotics will grow beyond the levels that were previously seen. This is mostly because continuous intake of antibiotics causes a monotonic decrease in the number of immune cells that are present inside the body.

\section{Numerical Simulations and Illustrations}

This section provides a visual representation of the dynamics of antibiotic uptake in the human body from different perspectives, allowing for improved awareness of how resistant bacteria harm our immune system day by day.\\

The numerical simulations of the proposed model have been accomplished with the help of MATLAB (R2018b). The fourth-order Runge-Kutta algorithm was implemented to show a visualisation of the relation between state variables and model parameters, as well as the illustration of phase portraits for the pertinent variables. The values of the parameters for the simulation have been presented in Table \ref{table:paramters value}.
  \begin{table}[hbtp!]
  \centering
   \caption{Parameter value and their descriptions}
    \label{table4.4.1}
      \begin{tabular}{p{1.4cm} p{6.3cm} p{2.2cm} p{1.0cm} p{1.2cm} }
      \hline
      \textbf{Symbol} & \textbf{Description} & \textbf{Unit} & \textbf{Value}& \textbf{Source}\\
      \hline
      $\beta_s$ & Growth rate of sensitive bacteria & $day^{-1}$ & 0.8& \cite{Mondragón2014}\\
      $\beta_r$ & Growth rate of resistant bacteria & $day^{-1}$ & 0.3&\cite{Daşbaşı2016}\\
      $\beta_m$ & Growth rate of immune cells & $day^{-1}$ & 0.6&\cite{Daşbaşı2016}\\
     $\eta$ & Bacteria eradication rate by immune cells & $day^{-1}$ & 0.3& \cite{khananik2020}\\
      $\omega$ & Rate of amount of present bacteria of carrying capacity of immune cells& $day^{-1}$ & 0.3&\cite{Daşbaşı2016}\\
     $\alpha$ & Mutation rate of sensitive bacteria & $mut \times gen$ & $0.3$& \cite{Daşbaşı2016}\\
       $\sigma$ & Elimination rate of sensitive bacteria & $day^{-1}$ & $0.0039$& \cite{khananik2020}\\
       $\delta$ & Dosage of antibiotic daily
 & $mgkg^{-1}day^{-1}$ & $0.0039$&\\
  $\mu$ & uptake rate of antibiotic
 & $day^{-1}$ & $0.06$& \cite{Daşbaşı2016}\\
 
$C$ & Carrying capacity of bacteria
 &  & $10^{9}$& \cite{Mondragón2014}\\
      \hline
    \end{tabular}
    \label{table:paramters value}
    \end{table}

In most common situations, antibiotics are typically prescribed for a duration of 14 days for various bacterial diseases. Patients with less severe infections are typically prescribed antibiotics for 5-7 days, while patients with more severe infections are prescribed antibiotics for 3-4 weeks, depending on the results of lab tests. Since the majority of patients are required to take antibiotics for 2 weeks, or 14 days, we have numerically solved the model for this duration and assumed that the patient would take 3 doses of antibiotics per day, with each dose administered at an 8-hour interval.

Furthermore, it has been considered that at the beginning of the antibiotic course, $a(0)=0$, the initial values of dimensionless variables sensitive bacteria (s), resistant bacteria (r), and immune cells in the body (m) are $s(0)=0.8,r(0)=0.1$ and $m(0)=0.3$, respectively.

There is a wide range of dosage strengths for antibiotics sold in pharmacies; for example, 5 mg, 50 mg, or 500 mg. Depending on the patient's age, health, and infection quality, the doctor can decide which one is best for that patient.  The model has been solved for a lower power antibiotic dosage here, and then a greater power dose would follow. An increased value of the antibiotic consumption rate $\mu$ is the consequence of a high value when the dosage of the antibiotic is boosted.
\begin{figure}[hbt!]
	\centering
	\includegraphics[height=6.3cm,width=10.8cm]{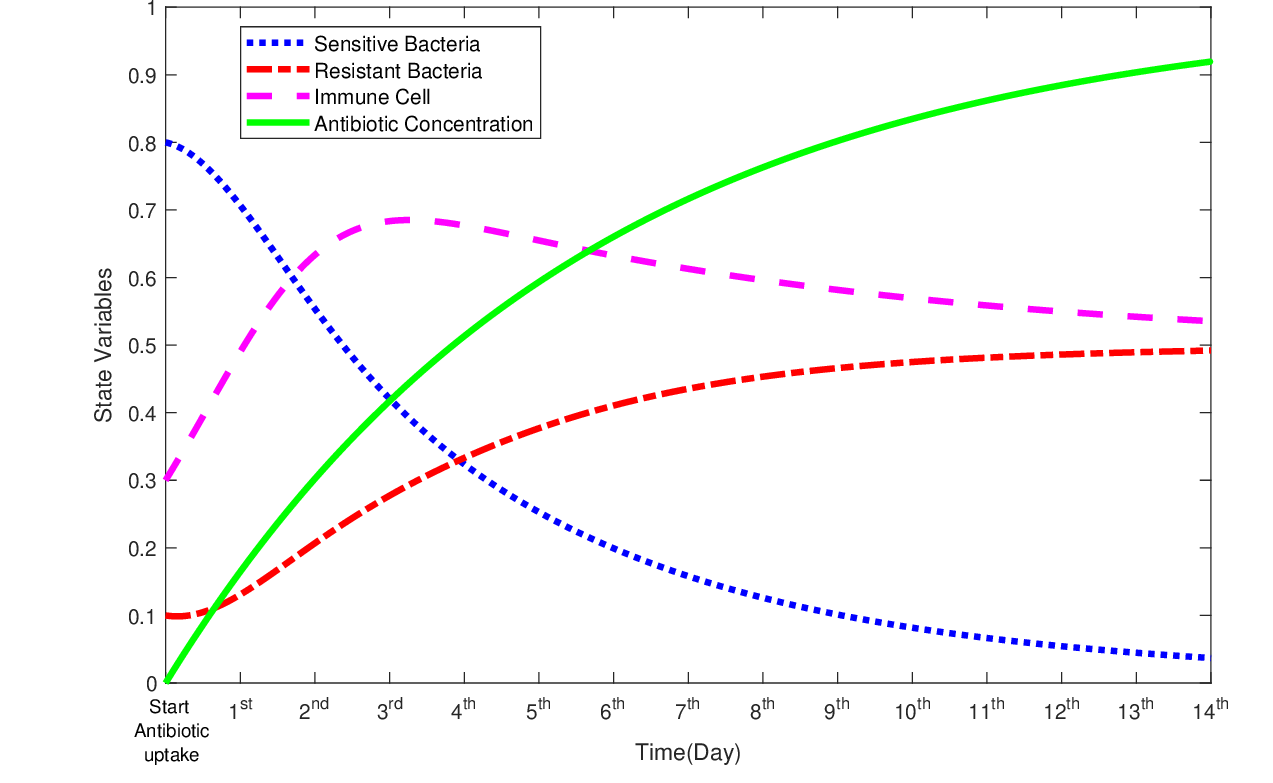}
	\caption{Dynamics of sensitive bacteria, resistant bacteria, immune cells, and antibiotic concentration up to the first $14^{th}$ days at antibiotic intake rate $\mu=0.06$  $day^{-1}$}.
	\label{fig: trajectory mu = 0.06}
\end{figure}

Now, the Figure \ref{fig: trajectory mu = 0.06} illustrates the behaviour of several state variables under the influence of a lower antibiotic dosage, characterised by an absorption rate of $\mu=0.06$ $day^{-1}$.

From this figure, it is evident that antibiotic concentration is increasing from 1-14 days with respect to time, and sensitive bacteria are decreasing as a result of the antibiotic effect. Because antibiotic kills sensitive bacteria with time. As we have used the antibiotic of less power in this figure, so the number of sensitive bacteria is decreasing slowly. At the same time, from this figure, we observed that resistant bacteria are increasing with time as the antibiotic concentration is increasing. Antibiotics are known to be ineffective against bacteria that have developed resistance. Additionally, it will lead to the evolution of bacteria that are resistant to previous strains. There has been an upsurge in the number of bacteria that are resistant to these treatments.

On the contrary, the amount of immune cells first rises to defend the body against bacterial invasion, reaching a high of 0.685 at around the fourth day of taking antibiotics.  Following the peak point, there has been a progressive decline in the number of immune cells, although at a slow rate. By the conclusion of the antibiotic treatment time, the remaining amount of antibiotic would be reduced to 0.53. Furthermore, it has been documented that the concentration of antibiotics in the body enhances simultaneously throughout the course of treatment, reaching a level of 0.919 during a period of 14 days.\\

In the meantime, with a higher antibiotic dose-based absorption rate $\mu=0.12$, the dynamical change of sensitive bacteria, resistant bacteria, immune cells, and antibiotic concentration are depicted in Figure \ref{fig: trajectory mu = 0.12}. It is apparent from this picture that the concentration of the antibiotic would grow rapidly at higher doses, as indicated. Therefore, the number of sensitive bacteria in our bodies will decrease more rapidly than shown in Figure \ref{fig: trajectory mu = 0.06} because using higher-strength antibiotics would kill these bacteria faster than before. Meanwhile, as shown in Figure \ref{fig: trajectory mu = 0.12}, the number of resistant bacteria is on the rise compared to when a lower dosage of antibiotics was used. Adding more dosages speeds up the mutation process by which bacteria that were formerly sensitive become resistant. Also, the immune cell count drops even more than it was following the lower dosage of antibiotic. This is because the beneficial bacteria already present in our colon are more quickly killed by stronger antibiotics. The rate of immune cells after the antibiotic period is 0.51 nearly in Figure \ref{fig: trajectory mu = 0.12}, compared to 0.54 almost in Figure \ref{fig: trajectory mu = 0.06}. 
\begin{figure}[hbt!]
	\centering
	\includegraphics[height=8.2cm,width=11.5cm]{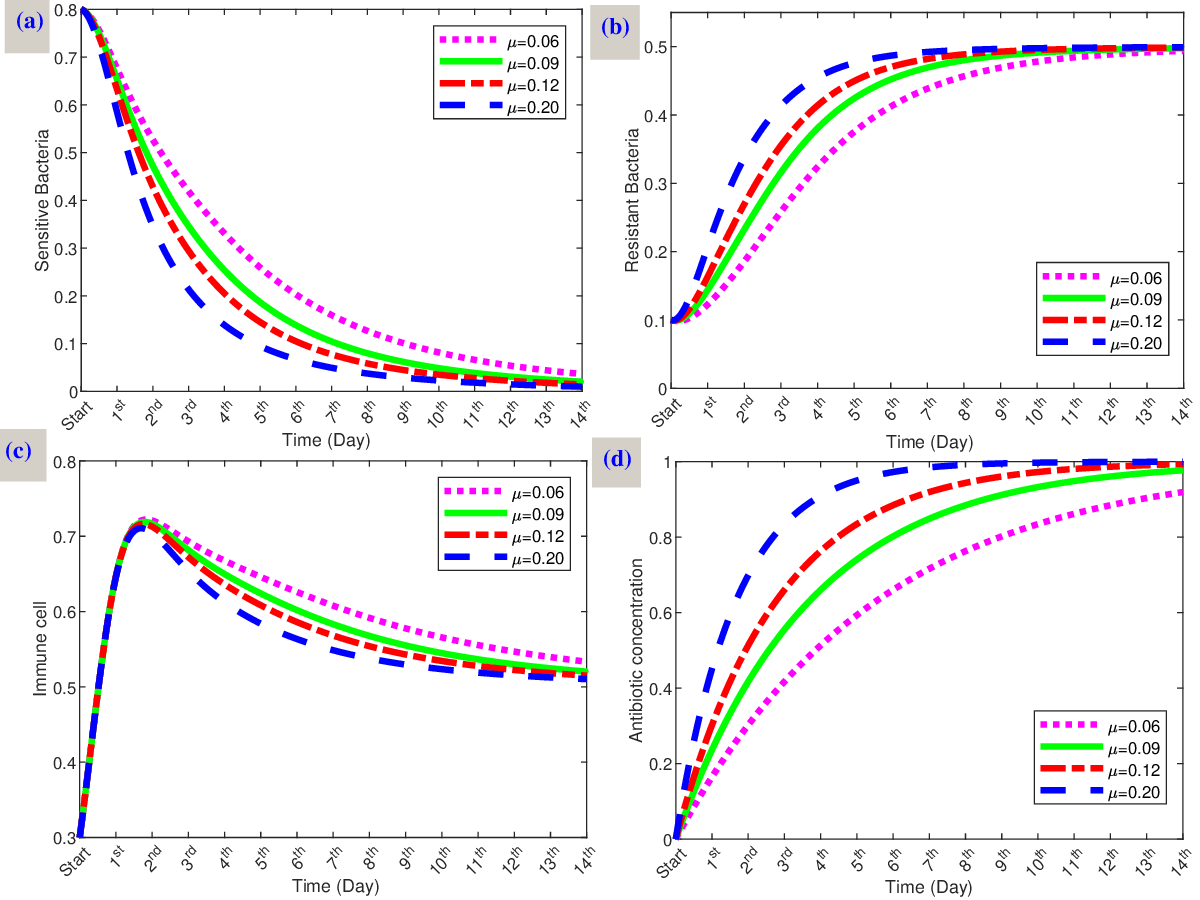}
	\caption{ (a) The larger the uptake rate through higher doss of antibiotics, the more sensitive bacteria are reduced. (b) More antibiotics are ingested in, leading to an increase in antibiotic-resistant bacteria. (c) Despite there may not be any noticeable change within a few days but increased rate of antibiotic absorption results in a greater reduction of immune cells.  (d) The concentration of antibiotics in the body rises as a result of taking antibiotics at higher doses.}
	\label{fig: mu variation}
\end{figure}

Concurrently, the dynamical variation in the state variables about $\mu$ are displayed in Figure \ref{fig: mu variation} where $\mu$ denotes the antibiotic uptake rate, and its value signifies the dosage concentrations of the antibiotics.
The sensitive bacteria number has started to decline for all $\mu$ values, as shown in Figure \ref{fig: mu variation} (a),  but the largest drop has been seen at $\mu=0.20$, which is much lower than the other lower values of $\mu$. Following the completion of 14 days of antibiotic treatment, it was also noted that the drug-sensitive bacteria's $\mu$ value steadily decreased to almost zero. \\
In contrast, Figure \ref{fig: mu variation}(b) depicts that initially, for all values of $\mu$, resistant bacteria have been increasing. However, the rate of boosting was very quick when $\mu=0.20$, and it only takes $5$ days for resistant bacteria to reach $0.5$, while other values reach their highest levels after $6^th$ day. Once again, looking at Figure \ref{fig: mu variation}(c), it is clear that none of the $\mu$ values significantly affect the immune cells, which reach a high of $0.715$ within $2^{nd}$ day. Interestingly, a significant change has been seen for immune cell reductions relative to $\mu$ after passing the peak point. Immune cells degrade more slowly when $\mu$ is small. Lastly, Figure \ref{fig: mu variation} (d) demonstrates a significant increase in antibiotic concentration as $\mu$ increases. At the maximum uptake rate of $\mu=0.20$, it turns out that the concentration remains constant after the seventh day, but other values of $\mu$ need more time to reach this consistent status.

\begin{figure}[hbt!]
	\centering
\includegraphics[height=6.3cm,width=10.8cm]{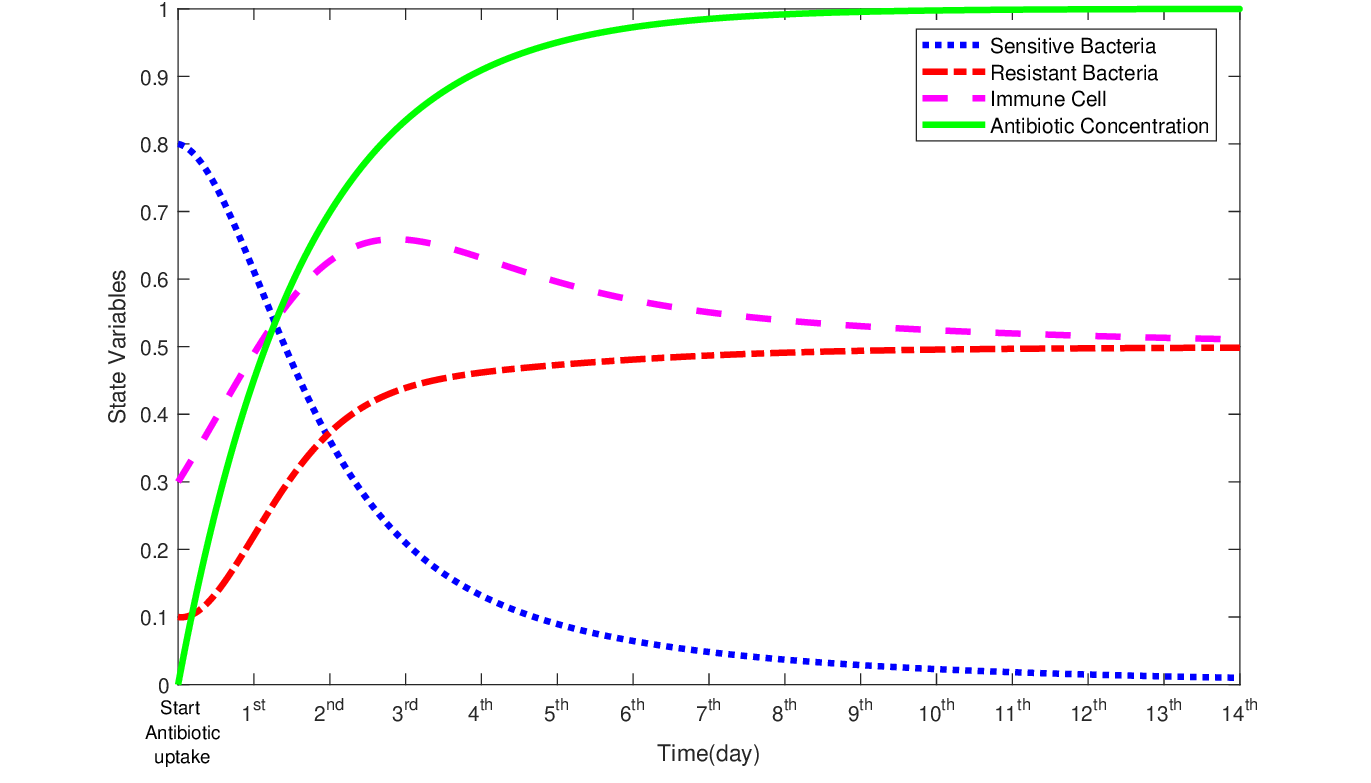}
	\caption{At a greater antibiotic intake rate of $\mu=0.12$ $day^{-1}$, susceptible and resistant bacteria, immune cells, and antibiotic concentrations varied over the first 14 days.}
	\label{fig: trajectory mu = 0.12}
\end{figure}
\begin{figure}[hbt!]
	\centering
	\includegraphics[height=8.4cm,width=11.3cm]{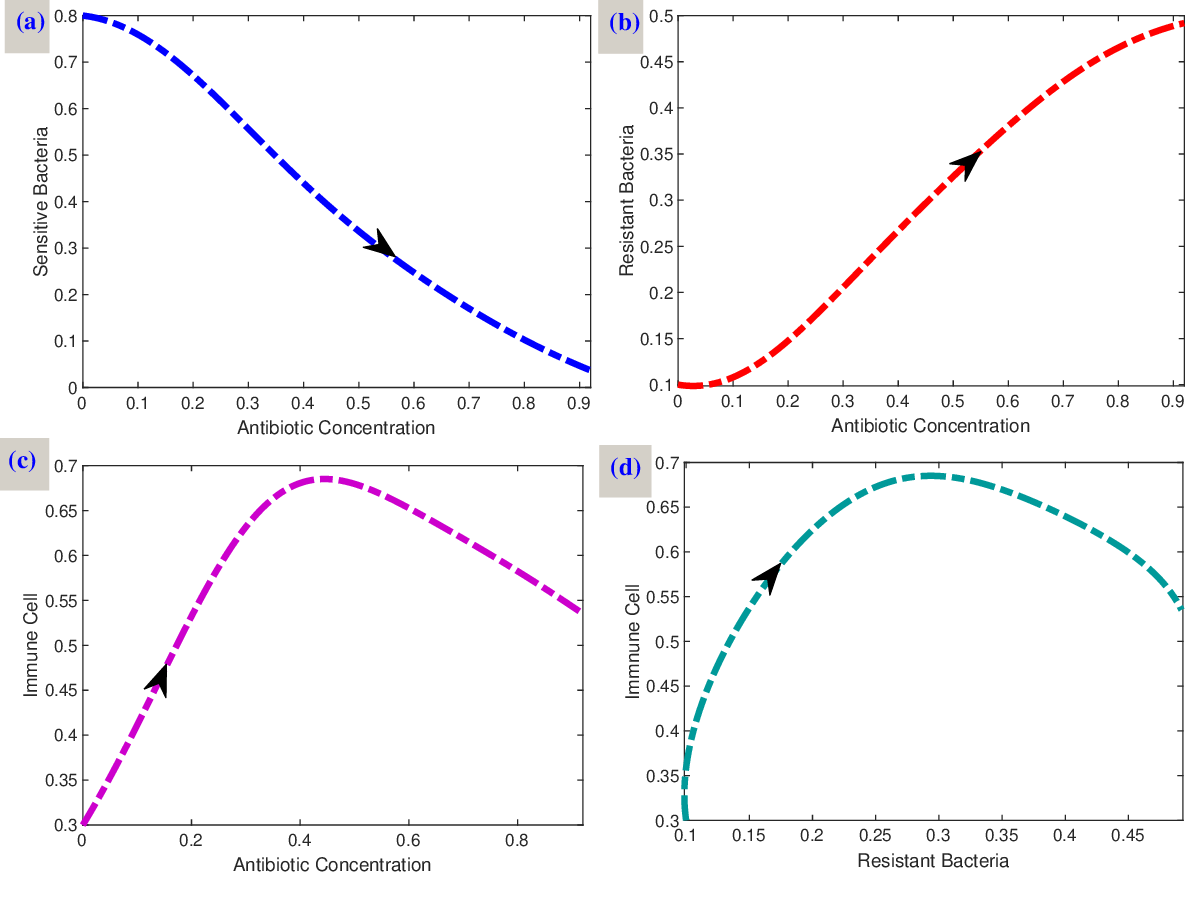}
	\caption{ (a) Even while the amount of sensitive bacteria decreases at first owing to the lower concentration of antibiotics, the rate of reduction has accelerated due to the increased concentration of antibiotics.
 (b) An enhanced concentration of antibiotics in the body leads to a rise in the population of antibiotic-resistant bacteria. (c) Antibiotics can enhance immune cells at first, though they eventually suppress them by destroying beneficial bacteria and encouraging the growth of antibiotic-resistant strains. (d) A higher value  of resistant bacteria at first makes immune cells go up, but then they start to go down because more resistant bacteria make it challenging for immune cells growth.}
	\label{fig: phase diagram general}
\end{figure}
Besides, the phase diagram shown in Figure \ref{fig: phase diagram general} reveals the qualitative dynamics of our model via simulations, particularly the relationship between various variables. Figure \ref{fig: phase diagram general}(a) illustrates the inverse relationship between sensitive bacteria and the levels of antibiotics in the body. From an antibiotic concentration of 0 to 0.1, there is a little drop in sensitive bacteria. However, as the antibiotic dosage exceeds 0.1, the sensitive bacteria decisively decline and finally a value of 0.037 when the antibiotic concentration grows to 0.91.
On the reverse side, Figure \ref{fig: phase diagram general}(b) displays a contrasting scenario to the preceding picture, where there is a direct relation between antibiotic concentration and resistant bacteria. In this case, the abundance of antibiotic concentration corresponds directly with the emergence of antibiotic-resistant bacteria in the body. Afterward, Figure \ref{fig: phase diagram general}(c) demonstrates an upward trend in immune cells in response to the escalating antibiotic concentration When antibiotic concentration is around 0.45, the peak of immune cells is seen. However, after time situation, the immune cells have declined despite the rising level of antibiotics. This is because antibiotics not only kill disease sensitive bacteria but also eliminate beneficial bacteria that enhance the immune system.  
Besides, a similar trend has been viewed about the connection between immune cells and resistant bacteria in Figure \ref{fig: phase diagram general} (d), whose pattern is similar to Figure \ref{fig: phase diagram general} (c).

Furthermore, the body's immune system plays an essential part in recovering from illness. Consequently, the impact of changes in various immune system situations on the model's different variables  has to be investigated. In this track, Figure \ref{fig: beta variation} displays the fluctuation of the various immune cell growth rates that characterize the status of the immune system, which may be weak or powerful. 
Before $3^{th}$ days, a lower value of $\beta=0.2$ resulted in reduced reductions of sensitive bacteria depicted in Figure \ref{fig: beta variation} (a), but the situation has turned around after this point. Due to the lower $\beta_m$ values, the sensitive bacteria have been drastically reduced in this situation. However, it is remarkable that a lower value of $\beta$ leads to a larger abundance of remaining sensitive bacteria in the end outcome.

\begin{figure}[hbt!]
	\centering	\includegraphics[height=8.2cm,width=11.5cm]{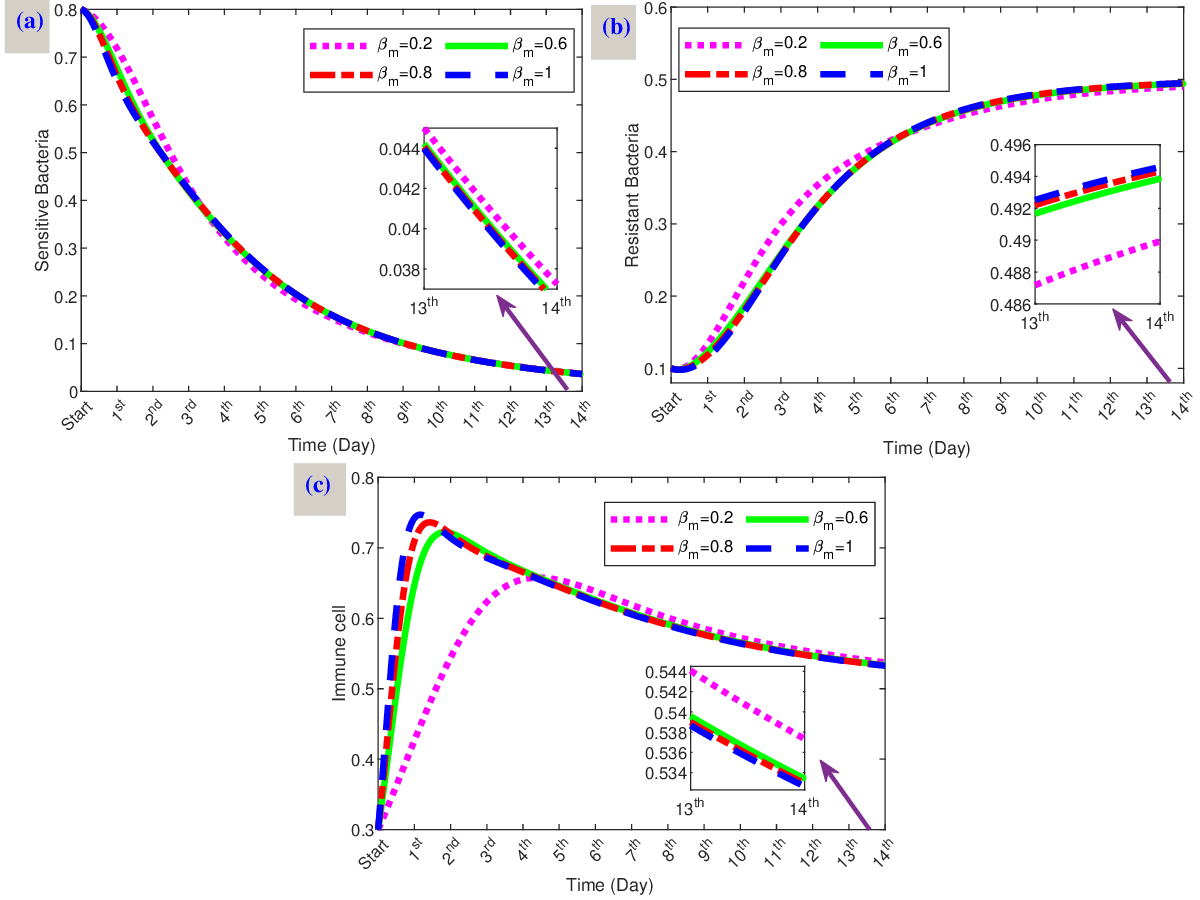}
	\caption{(a) The decline of sensitive bacteria is facilitated by a higher rate of immune cell growth. (b) Slower immune cell growth at first enhances the number of resistant bacteria, but as time goes on, faster growth rate shows that there are more resistant bacteria because of the accumulation of antibiotics. (c) An rise in the growth rate of immune cells initially results in a greater quantity of immune cells, but as time passes, the reverse scenario is seen. }
	\label{fig: beta variation}
\end{figure}

Figure \ref{fig: beta variation} (b) points out that for the lower value of $\beta=0.2$, the enhancement of resistant bacteria is more pronounced up to the sixth day, whereas there are no significant differences seen for other values of $\beta_m$. But over time, this scenario has faded. Interestingly, towards the conclusion of the antibiotic-taking time, resistance in bacteria is also larger when $\beta_m$ is higher. Therefore, an increased rate of immune cell growth does not always result in a reduction in the presence of resistant bacteria due to the antibiotic effect.

Again, the data in Figure \ref{fig: beta variation} (c) demonstrates that the immune cells exhibit a rapid increase and reach a greater peak when the value of $\beta_m$ is bigger, specifically for $\beta_m=1$, compared to the values of $\beta_m=0.8$ and $\beta_m=0.6$. Subsequently, the immune cell has reduced, whereas higher values of $\beta_m$ correspond to a greater reduction. On the other hand, when $\beta_m$ is at its lowest value of 0.2, the immune cell reaches its highest concentration on the $5^{th}$ day. However, the drop in level is not as major as before, maybe because there are fewer resistant bacteria present compared to when $\beta_m$ was greater. On the flip side, when $\beta_m=0.2$, the peak occurred at a slower rate, particularly on the fifth day of the prescription. At the end of the prescription, the level of immune cells remained higher compared to the other values of $\beta$.

When we execute our previous MATLAB code implementing the Simpson one-third rule to solve the integration problem defined in equation \ref{eq:Integration} using the parameters values $B_1=10, B_2=5$ and $B_3=5$ , it has come to light that the changes in the values of our objective function corresponding to the variations in $\mu$ (representing the antibiotic power) and $\beta_m$ (indicating the immune system strength). These variations are presented in Table \ref{table:integration values}. Based on the information presented in the table, it can be concluded that when $\mu=0.06$ and $\beta_m=1.0$, the value of our objective function reaches its maximum at $298.99$. Furthermore, when the value of $\mu$ increases, the cost of antibiotics along with the prevalence of resistant bacteria likewise increases but unfortunately the value of our objective function has fallen. Additionally, it has been noted that when the immune system is weakened, a greater dosage of antibiotics is unable to enhance our objective function values. Remarkably, even with a greater value of $\beta_m$, the effectiveness of stronger antibiotics will not enhance our objective function's values.
 \begin{table}[hbt!]
  \centering
   \caption{Functional value of equation (\ref{eq:Integration}) for variation of $\mu$ and $\beta_m$ when $B_1=10, B_2=5$ and $B_3=5$}
    \begin{tabular}{c|c|c|c|c|c}
      \backslashbox{\textbf{$\beta_m\downarrow$}}{$\mu\rightarrow$}& \textbf{$\mu=0.06$} & \textbf{$\mu=0.09$} & \textbf{$\mu=0.12$} & \textbf{$\mu=0.15$}& \textbf{$\mu=0.20$}\\
      \hline
      $\beta_m=0.2$ & 266.39& 225.49&201.86&186.79 & 171.16\\
      \hline
       $\beta_m=0.4$ & 286.64& 244.99& 220.89&205.47 & 171.17\\
      \hline
       $\beta_m=0.6$ & 293.48& 251.69& 227.52& 212.06 & 195.96\\
      \hline
       $\beta_m=0.8$ & 296.92& 255.07& 230.89& 215.42 & 199.30\\
      \hline
       $\beta_m=1.0$ & \textcolor{blue}{298.99}& 257.11& 232.93& 217.45 & 201.32\\
    \end{tabular}
    \label{table:integration values}
    \end{table}
Moreover, it is normally observed that people often discontinue taking antibiotics just before few days of the doctor's recommended period or keep taking them for a few days after the specified period has passed. Here, our proposed floor function (see in equations (\ref{equ:under_model equation1}- \ref{equ:under_model equation4}) and (\ref{equ:over_model equation1}-\ref{equ:over_model equation4})) provide similar outcomes in scenarios where there are instances of underdosing and overdosing. But, if the prescribed time frame is 14 days and the interval for taking is either fewer than 7 days or more than 28 days, then this floor function won't provide the desired outcome. As a result, our introduction of the new floor function has shown the advantage of describing circumstances mathematically within a shorter period (as illustrated in subsecti \ref{se:under} and \ref{se:over}) from the prescribed time interval.

\subsection{Effect for Under-dose Treatment}\label{se:under}

The majority of people have a common tendency to fulfil the recommended antibiotic dosages after they think that their health has improved and they are no longer experiencing symptoms of the illness. When a patient discontinues the use of antibiotics before completing the whole course, as advised by the doctor, it is referred to as the underdosing of antibiotics. Assuming that a doctor prescribes a 14-day course of antibiotics. If the patient takes antibiotics for a duration of fewer than 14 days, it is referred to as an under-dose of antibiotics. However, this phenomenon poses a serious threat to human health, which is illustrated graphically in this part through the introduction of a new floor function into the preceding model.\\
\begin{figure}[hbt!]
	\centering
	\includegraphics[height=6.3cm,width=11cm]{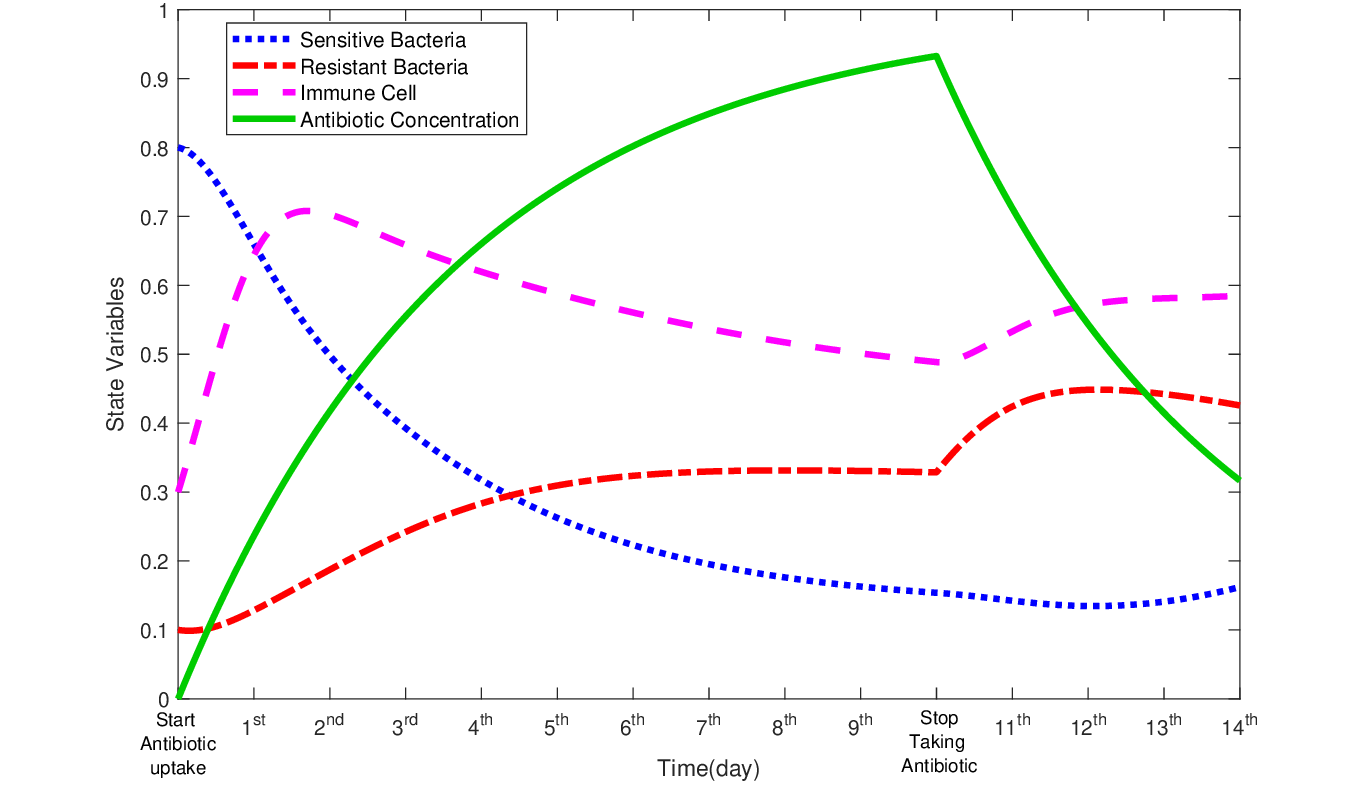}
	\caption{Trends in sensitive bacteria, resistant bacteria, immune cells, and antibiotic concentration during the first 14 days of taking an antibiotic at a rate of $\mu=0.06$ $day^{-1}$ when the prescription antibiotics are stopped 10 days before the end of the dose.}
	\label{fig: under-dose}
\end{figure}
In order to do this, we have made modifications to our model (\ref{equ:under_model equation1})-(\ref{equ:under_model equation4}) by incorporating a new function $\floor*{\frac{t}{T_d}}$ into the variables representing the change in resistant bacteria and antibiotic concentration. This function specifically emphasizes the impact of under-dosing antibiotics. According to this, our newly proposed model would resemble the following:
\begin{align}
    \frac{ds}{dt}&=\beta_s s\left(1-(s+r)\right)-\eta sm-\alpha sa-\sigma s a \label{equ:under_model equation1}\\
    \frac{dr}{dt}&=\beta_r r(1-(s+r))-\eta rm+\alpha sa +\kappa a \floor*{\frac{t}{T_d}} \label{equ:under_model equation2}\\
    \frac{dm}{dt}&=\beta_m m\left(1-\frac{m}{(s+r)}\right)\label{equ:under_model equation3}\\
    \frac{da}{dt}&=\mu\left(1-\floor*{\frac{t}{T_d}}\right)-\mu a\label{equ:under_model equation4}
\end{align}
To point out the effects of under-dosing in the treatment period, a \textit{floor function} is newly introduced and used as a \textit{switch function} in the model equations. For the under-dosing scenario, we set the treatment duration threshold as \( T_d = 10 \) days. The floor function is defined as $\left\lfloor \frac{t}{T_d} \right\rfloor$. This function behaves as a binary switch because
\begin{itemize}
    \item When \( t < T_d \), we have \( \left\lfloor \frac{t}{T_d} \right\rfloor=0 \left(i.e.  \left\lfloor \frac{5}{10} \right\rfloor = \left\lfloor 0.5 \right\rfloor = 0\right)\), so \( \left(1 - \left\lfloor \frac{t}{T_d} \right\rfloor \right) = 1 \), and the first term in equation \eqref{equ:under_model equation4} remains active.
    \item When \( t \geq T_d \), we have \( \left\lfloor\frac{t}{T_d} \right\rfloor = 1\left(i.e.  \left\lfloor \frac{12}{10} \right\rfloor = \left\lfloor 1.2 \right\rfloor = 1\right)\), so \( \left(1 - \left\lfloor \frac{t}{T_d} \right\rfloor \right) = 0 \), which effectively \textit{switches off} the first term in equation \eqref{equ:under_model equation4}.
\end{itemize}
Similarly, in equation \eqref{equ:under_model equation2}, the term \( \kappa a \) becomes active after 10 days since
\[
\left\lfloor \frac{t}{T_d} \right\rfloor = 1 \quad \text{for} \quad t \geq 10.
\]
Hence, the floor function is utilized as a \textit{switch function} to activate or deactivate certain terms in the model equations depending on whether the time \( t \) has crossed the treatment threshold \( T_d \). This structure allows the model to realistically reflect the consequences of under-dosing (premature termination of treatment) or over-dosing by dynamically controlling term activation. Here, Figure \ref{fig: under-dose} displays the model's behaviour. Using the above concept, the fourth-order Runge–Kutta method (method for solving IVP) was employed to obtain the results shown in Figure \ref{fig: under-dose}.\\ 
\begin{figure}[hbt!]
	\centering
	\includegraphics[height=8.4cm,width=11.3cm]{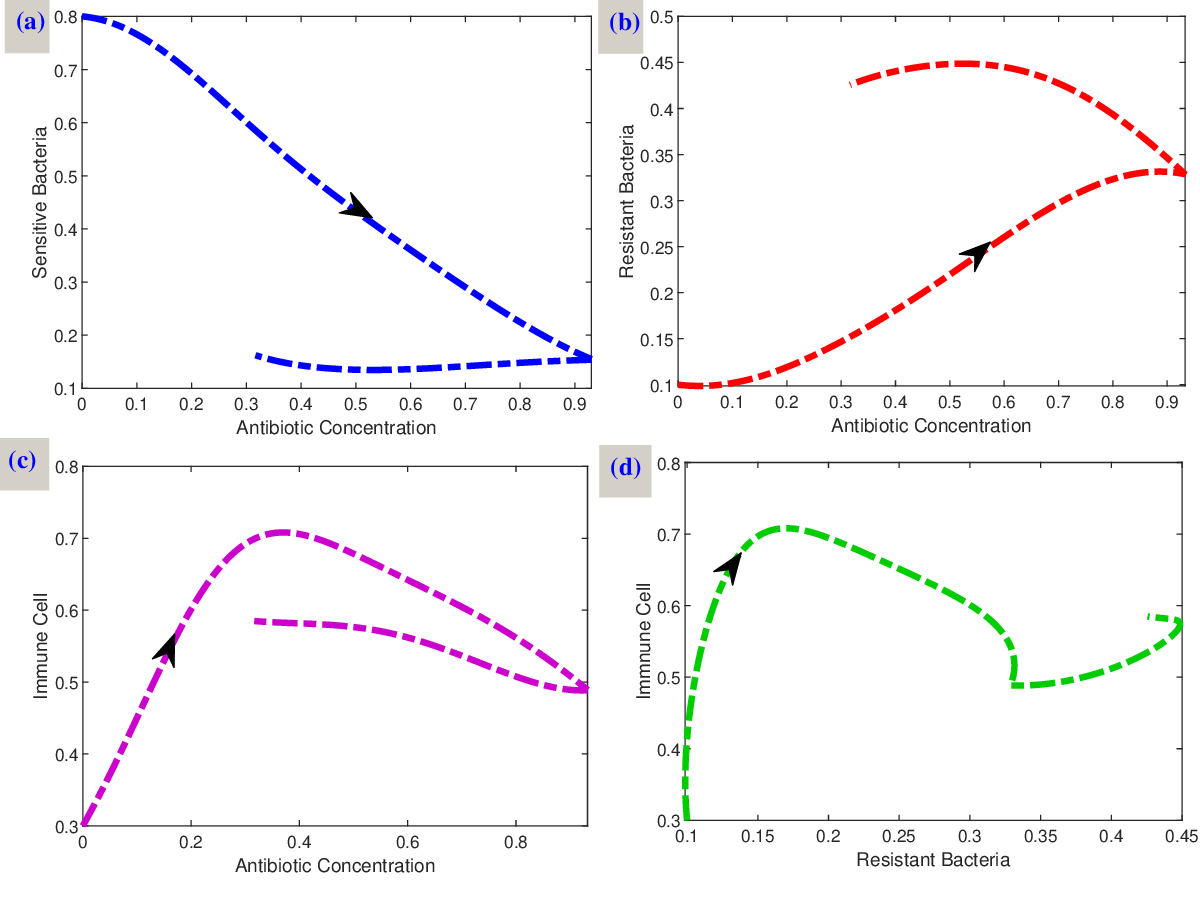}
	\caption{(a) Increased antibiotic concentration in the body reduces sensitive microorganisms, whereas reduced concentrations show the reverse.
 (b) More antibiotic concentration in the body gradually builds resistant bacteria, but stopping antibiotics under dosage promotes resistant bacteria more quickly. (c) Although taking antibiotics lower the number of immune cells up to the point where they don't kill useful bacteria at a larger amount, but stopping antibiotic use raises the number of immune cells. (d) Initially, immune cells increase in response to growing resistant bacteria, but this trend diminishes as the prevalence of resistant bacteria becomes more dominant.  }
	\label{fig: phase diagram underdose}
\end{figure}
Figure \ref{fig: under-dose} points out an eye-catching change that occurs after discontinuing antibiotic use. This graph illustrates the gradual reduction in antibiotic concentration in the body after the 10th day, resulting from the discontinuation of antibiotic intake. Besides, after 12 days, the presence of sensitive bacteria in the body has started to grow due to the lessening concentration of antibiotics. At the same time, immune cells have had an upward trend, reaching a value of $0.58$ at $14^{th}$ days. Unfortunately, at this time, the adaptation of sensitive bacteria to survive in environments with inadequate antibiotic concentrations has led to the growth of resistant bacteria at a rate greater than full dosage treatments.\\ 
Therefore, we may deduce that administering a lower dosage of antibiotics leads to a faster rate of mutation from sensitive to resistant bacteria, thereby resulting in a higher population of resistant bacteria inside the human body. All of the bacteria will become antibiotic-resistant if this condition persists. If that were the case, any antibiotic would not work against diseases caused by bacteria.\\
Furthermore, Figure \ref{fig: phase diagram underdose} illustrates the relationship between various variables in the context of an under-dose scenario of treatment with antibiotics. According to Figure \ref{fig: phase diagram underdose}(a), when the body's concentration of antibiotics rises, the number of susceptible bacteria drops from $0.8$ to $0.15$. After that, sensitive bacteria steadily increase in concentration when antibiotics are discontinued after ten days of intake. Next, Figure \ref{fig: phase diagram underdose} reflects the rise of resistant bacteria once the antibiotic concentration in the body reaches $0.1$. This increasing pattern continues until the resistant bacteria has reached a level of $3.34$. After that, when antibiotics were stopped being taken before the end of the prescription time, the number of bacteria that were resistant to the antibiotics rose once again, and by the time the prescription•time•was over, there was a slight decline in the number of resistant bacteria. 
In contrast, Figure \ref{fig: phase diagram underdose} (c) indicates that on the $10^{th}$ day after stopping the antibiotic, the remaining concentration of the antibiotic in the body was 0.93, while the immune cell was 0.49. Interestingly, by the end of the recommended period, although the antibiotic concentration was declining, the immune cell measure had risen to 0.58.  At the same time, immune cells rise in addition to the increase of resistant bacteria when antibiotics are stopped before a dosage is finished, although both eventually decrease, as can be observed in Figure \ref{fig: phase diagram underdose} (d). Thus, Figure \ref{fig: phase diagram underdose} has been drawn to explore the short-term phase dynamics between antibiotic concentration and the other key variables. It highlights that under-dosing reduces the effectiveness of treatment, allowing sensitive bacteria to survive and mutate into resistant strains, while only slightly preserving immune cells. The diagram confirms that early discontinuation of antibiotics can lead to a sharp rise in resistance.

\subsection{Effect for Over-dose Treatment}\label{se:over}

When patients have not fully recovered from the prescribed antibiotic dosage, many of them want to continue taking the antibiotic without any consultation with their doctor. In this sense, over-dose of antibiotics refers to the act of consuming a greater quantity of antibiotics than what has been specifically recommended by a medical professional. Assume that a physician prescribes a course of antibiotics for a duration of 14 days. An over-dose of antibiotics occurs when the patient exceeds a duration of 14 days while taking the medication. This part focuses on the dynamic evolution of the human health system due to the overdose effect of antibiotics, which is shown via graphical depiction.\\ 
\begin{figure}[hbt!]
	\centering
	\includegraphics[height=6.3cm,width=10.8cm]{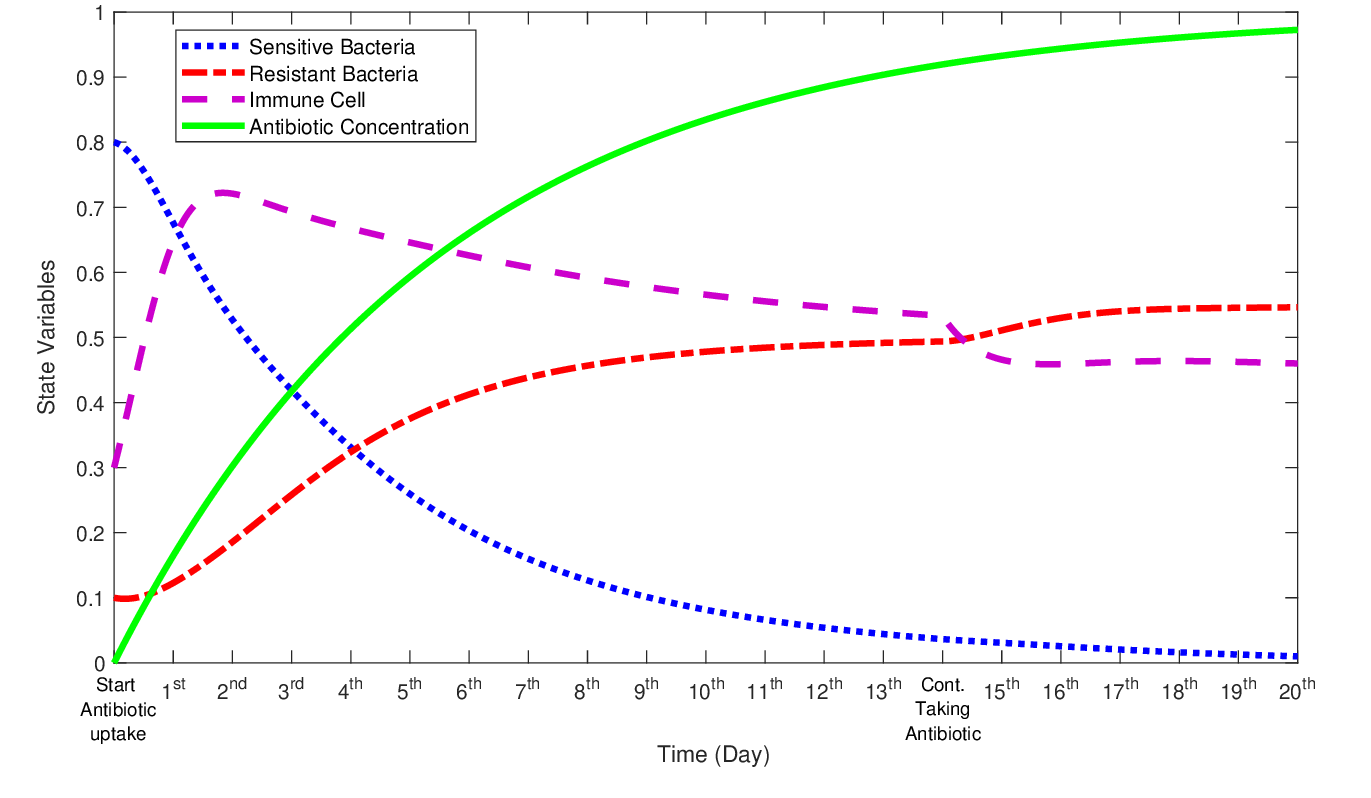}
	\caption{ Patterns of sensitive bacteria, resistant bacteria, immune cells, and antibiotic concentration throughout the first 20 days of treatment at a rate of $\mu=0.06$ $day^{-1}$ when antibiotics have been taken continuously over the 14-day prescription time interval.}
	\label{fig: Over-dose}
\end{figure}
Now we have modified our model (\ref{equ:over_model equation1})-(\ref{equ:over_model equation4}) for over-dose of antibiotics. Let us consider that over-dose of antibiotics decreases the immune cells at a rate. So we get our modified model as,
\begin{align}
    \frac{ds}{dt}&=\beta_s s\left(1-(s+r)\right)-\eta sm-\alpha sa-\sigma s a \label{equ:over_model equation1}\\
    \frac{dr}{dt}&=\beta_r r(1-(s+r))-\eta rm+\alpha sa \label{equ:over_model equation2}\\
    \frac{dm}{dt}&=\beta_m m\left(1-\frac{m}{(s+r)}\right)-\phi a \floor*{\frac{t}{T_D}}\label{equ:over_model equation3}\\
    \frac{da}{dt}&=\mu(1-a)\label{equ:over_model equation4}
\end{align}
\begin{figure}[hbt!]
	\centering
	\includegraphics[height=8.4cm,width=11.3cm]{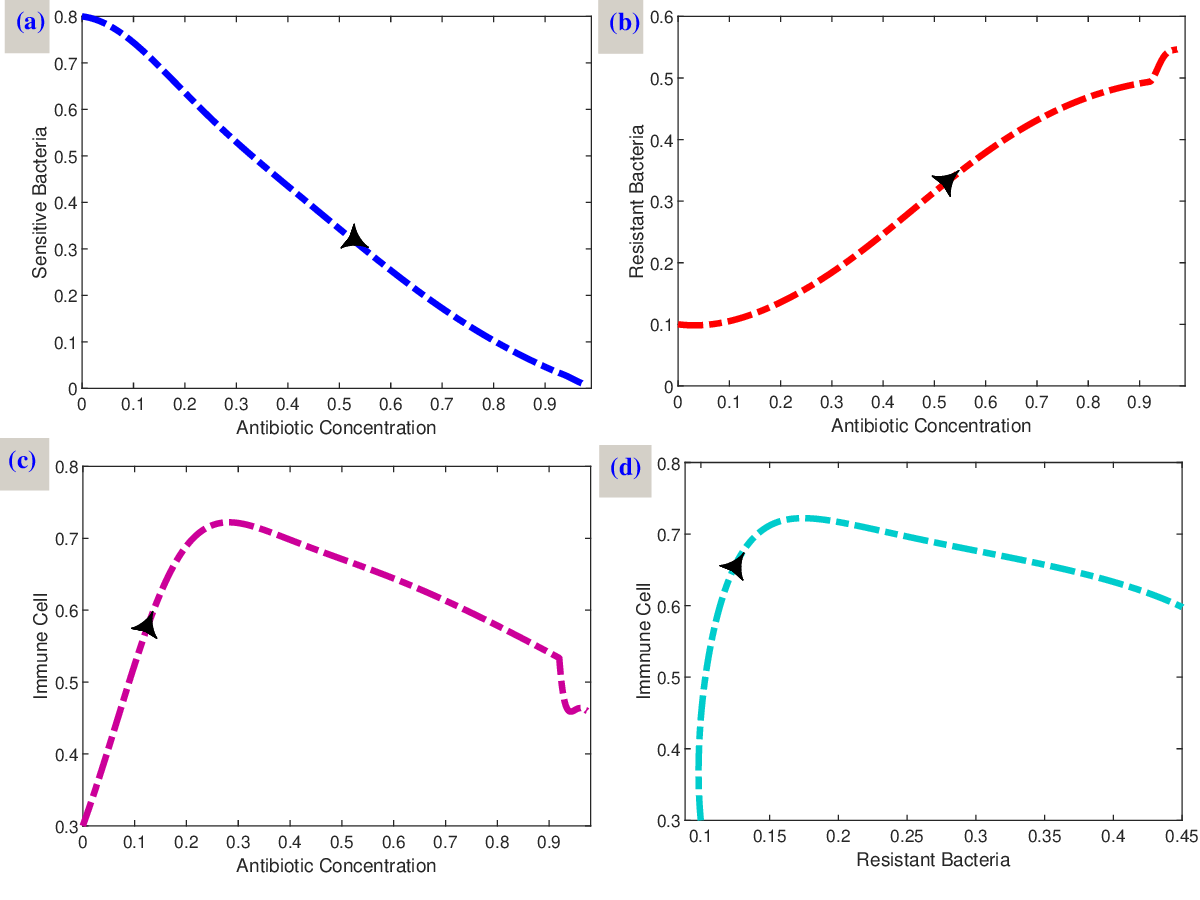}
	\caption{(a) The increasing concentration of antibiotics has led to a continuous decline in sensitive bacteria. (b) Bacterial resistance arises over time as antibiotic concentration grows, but it accelerates with overdose.
 (c) A relationship between antibiotic concentration and immune cell number has been displayed. (d) Antibiotic's overuse causes immune cells to initially an increase in response to the rise of resistant bacteria, but excessive use of these drugs is detrimental to the immune system.}
	\label{fig: phase diagram overdose}
\end{figure}
For the equation (\ref{equ:over_model equation3}), 14 days has been considered as a prescribed period of antibiotics by the doctor, but patients have taken the dosage for as long as $20^{th}$ days without medical supervision. Based on this, we have analyzed the model (\ref{equ:over_model equation1}- \ref{equ:over_model equation4}) numerically with the help of Runge-kutta fourth order taking all the values of the parameters being same as in Table \ref{table:paramters value} using the following techniques:
Let, \( T_D = 14 \). In this case:
\begin{itemize}
    \item When \( t < 14 \), we have \( \left\lfloor \frac{t}{T_D} \right\rfloor = 0 \), which implies \( \phi a \left\lfloor \frac{t}{T_D} \right\rfloor = 0 \). Thus, the corresponding term in equation \eqref{equ:over_model equation3} remains inactive—functioning like a switch-off mechanism.
    \item When \( t \geq 14 \), \( \left\lfloor \frac{t}{T_D} \right\rfloor = 1 \), activating the term \( \phi a \left\lfloor \frac{t}{T_D} \right\rfloor \) and thereby showing its effect within equation \eqref{equ:over_model equation3}.
\end{itemize}
Thus, for the over-dosing case as well, the floor function acts as a switch, activating or deactivating terms based on whether the treatment period has exceeded the prescribed duration.\\
From Figure \ref{fig: Over-dose}, we observed that after $14^{th}$ days when an over-dose of antibiotics is consumed (taking the antibiotic administration rate maintained active for 20 days), the immune system is started to fall down. At the same time, the number of resistant bacteria increased at a considerably higher rate.\\
Now, the relationship between the various state variables in the model when the overdose of antibiotics occurred has been brought to light in Figure \ref{fig: phase diagram overdose}.  The presence of the additional quantity of antibiotics concentration caused by over-use has resulted in a higher reduction in sensitive bacteria, as shown in Figure \ref{fig: phase diagram overdose}(a). Contrarily, Figure \ref{fig: phase diagram overdose}(b) indicates that the population of resistant bacteria grows gradually until the antibiotic concentration reaches 0.93. However, when the antibiotic concentration surpasses this level owing to excessive use, the population of resistant bacteria decreases drastically with further increases in antibiotic concentration. Also, it has been observed that the higher antibiotic concentration of overdose reduces more immune cells within a shorter rise of concentration and resistant bacteria (see Figures \ref{fig: phase diagram overdose}(c) and (d)). 
Overall, Figure 12 illustrates the effects of antibiotic overdose through phase-plane analysis. The results show that excessive antibiotic administration significantly lowers immune cell levels due to the destruction of beneficial bacteria, while simultaneously promoting the dominance of resistant bacteria. Therefore, we can infer that overuse of antibiotics leads to an escalation in the population of antibiotic-resistant bacteria inside the human body, posing a serious threat to the potency of antibiotics in combating bacterial infections.

\section{Results and Discussion}
This study analyzes the effect of varying antibiotic dosages on bacterial resistance using our proposed model. As shown in Figures~\ref{fig: trajectory mu = 0.06}--\ref{fig: trajectory mu = 0.12}, high-strength antibiotics increase resistant bacteria and reduce immune cells by supporting resistant strains and disrupting beneficial microbiota. Figure~\ref{fig: phase diagram general} further reveals a clear relationship between antibiotic concentration and the populations of both sensitive and resistant bacteria. Figure~\ref{fig: beta variation} shows that resistant bacteria increase more slowly and stabilize at lower levels when immune cell growth and abundance are higher, highlighting the immune system's crucial role in controlling bacterial resistance. Furthermore, according to the evidence shown in Table \ref{table:integration values}, it reveals that by enhancing our immune system by taking of a lower dosage of antibiotics, we can successfully achieve our desired target for our objective function.  
On the other hand, according to the information provided in Figures \ref{fig: under-dose} and \ref{fig: phase diagram underdose}, under-dosage of antibiotics might lead to many health issues in the human body. Insufficient dosage of antibiotics reduces the rate at which sensitive bacteria are killed. The primary issue associated with under-dosing antibiotics is the development of genetic mutations in sensitive bacteria. Stopping antibiotic treatment prematurely greatly increases the chance of leaving behind a substantial number of sensitive bacteria that were being targeted by the prescription antibiotic. As a result, these sensitive bacteria will rapidly undergo genetic mutations, leading to the emergence of resistant bacteria. So, under-dose of antibiotics will elevate the mutation rate, thereby boosting the population of antibiotic-resistant bacteria. Conversely, the immune system becomes a bit better as a result of a reduced death rate of beneficial bacteria caused by a lower dosage of antibiotics as these bacteria are essential for the production of immune cells. This finding for under-dose treatment aligns with the previously performed experimental investigation conducted by Llewelyn et al. (2017) \cite{Llewelynj3418}. Besides, Figure \ref{fig: Over-dose} as well as Figure \ref{fig: phase diagram overdose} demonstrate that a significant drop in human immune cells occurs when antibiotics are overdosed. Antibiotics reduce the effectiveness of the immune system because they eliminate beneficial bacteria in the intestines. These bacteria are responsible for producing immune cells and enhancing the immune system's overall strength.  Overdosing on antibiotics, however, speeds up the process just similar to the study mentioned in Benabbou et al. (2023) \cite{Benabbou2023}. So, the human immune system loses its strength more significantly when antibiotics are overdosed. Normally, the immune system can eliminate bacteria that have developed resistance to antibiotics. However, if patients continue taking antibiotics beyond the prescribed dosage, the immune system weakens, allowing resistant bacteria to increase.\\
Therefore, this study provides insights that are directly relevant to antibiotic regulation and public health policy. The findings highlight the need for stricter control over antibiotic distribution and enforcement of physician-guided dosing to minimize under- and over-dosing, which are key drivers of resistance. Public awareness campaigns can complement these efforts by educating the population about the risks of self-medication and misuse. Moreover, integrating this model with real-time surveillance data could help policymakers evaluate and fine-tune intervention strategies. Clinically, the model can serve as a valuable decision-support tool for hospitals and health departments to simulate dosing outcomes and optimize treatment plans, contributing to more effective and sustainable antibiotic use. 

\section{Conclusions}
Antibiotic resistance is a global concern, and Bangladesh plays a significant role in contributing to it due to the long-standing problem of antibiotic usage in the country. In contrast to several countries, antibiotics may often be acquired at pharmacies in this location without the need for a doctor's prescription \cite{AHMED201954}. On the other hand, certain bacteria in our body are divinely bestowed, which protect our bodies. Unnecessary antibiotics are decimating these beneficial bacteria. Consequently, rather than experiencing the intended healing effects of the drug, we were experiencing illness. Annually, 170,000 individuals die in the nation as a result of antibiotic resistance. If this pattern continues, it is projected that by 2050, the number of individuals who die as a result of antibiotic resistance might exceed those who die from the coronavirus \cite{TBS2023}. Therefore, it is high time to establish a more effective antibiotic-use strategy at this moment.

In order to have a better understanding of these life-threatening challenges, a mathematical model has been developed that has been figured out analytically and numerically. Our theoretical investigation has shown the existence of a unique solution for the model, which can be assured by the well-posedness of the model. Out of the five equilibrium points in the model system, $E_4$ and $E_5$ exhibit asymptotic stability under certain conditions, whereas the other three points demonstrate unstable. In the meantime, an analysis has been conducted on the equilibrium properties of the model's parameter. It has also been found that when the antibiotic concentration decreases an ongoing decline in immune cells, resulting in a continual reduction of resistant bacteria. 

The first portion of our numerical simulations included observing the bodily reactions of sensitive bacteria, resistant bacteria, and the immune system to several levels of antibiotic intake, each with different strengths.  Antibiotics are used to combat diseases, however, bacteria use their mutation mechanism to resist antibiotics, resulting in the emergence of antibiotic-resistant bacteria. Our analysis indicates that this mutation is an ongoing phenomenon that cannot be halted, but it may be mitigated. Meanwhile, it would be advantageous to enhance the growth of the immune system by diminishing the therapeutic value of antibiotics in order to reduce the expenses associated with achieving our goal of attaining the highest level of beneficial immune cells. Nevertheless, using the appropriate strength of antibiotics may diminish the prevalence of this mutation. Furthermore, to illustrate the effects of underdosage and overdosage of antibiotics, a new floor function has been introduced. With the assistance of this function, it has been noted that even if immune cells have begun to grow during this period, greater development of resistant bacteria has occurred if the recommended course of antibiotics is not finished. However, if someone continues to use antibiotics that they purchased without a prescription, there is an enormous danger for the future treatment system since immune cells are waning and resistant bacteria have become more widespread. Thus, the major findings of this study are:
\begin{itemize}
    \item Excessive use of high-strength antibiotics increases the final population of resistant bacteria while weakening the immune system by disrupting beneficial bacteria.
    
    \item The newly proposed floor function in this study effectively captures the switching behaviour of antibiotic action for both over- and under-dosing scenarios.
    
    \item Under-dosing antibiotics leads to incomplete elimination of sensitive bacteria, increasing mutation risk and accelerating the emergence of resistant strains.
    
    \item Overuse of antibiotics after the prescribed period further suppresses immune cell production and promotes antibiotic resistance.
    
    \item A higher immune cell growth rate plays a protective role by slowing the accumulation of resistant bacteria and reducing their final abundance.
    
    \item A balanced strategy combining lower antibiotic dosage with enhanced immune response can prove effective in minimizing resistance and achieving treatment goals.
\end{itemize}
Therefore, it is vital to promptly adopt initiatives to raise public awareness for reducing the growth of antibiotic-resistant bacteria by limiting the use of antibiotics. Failure to do so would render future medical treatments ineffective, perhaps leading to fatal consequences even for minor illnesses. Future studies will integrate optimal control theory and cost-effectiveness approaches to effectively minimize resistance and design economically sustainable antibiotic treatments.
\section*{Data Availability Statement}
The data used to support the findings of this study are included in the article.
\section*{Conflicts of Interest}
The authors declare that they have no known financial or interpersonal conflicts that would have caused it to have an impact on the study described in this article.
\section*{Funding}
The research did not receive any specific funding.





\bibliography{sn-bibliography}

\end{document}